\documentclass{article}

\usepackage{fullpage}
\usepackage{amsmath,amsthm,amssymb,bm}
\usepackage{mleftright}

\usepackage{hyperref,url,xcolor}

\newcommand{\oA}{\bar{A}}

\newcommand{\MDS}{\operatorname{MDS}}

\newcommand{\GZP}{\operatorname{GZP}}
\newcommand{\MDSb}{\operatorname{MDSb}}

\newcommand{\LDMDS}{\operatorname{LD-MDS}}
\newcommand{\wt}{\operatorname{wt}}
\newcommand{\rank}{\operatorname{rank}}

\newcommand{\eps}{\epsilon}

\newcommand{\F}{\mathbb F}

\newcommand{\Z}{\mathbb Z}
\newcommand{\cS}{\mathcal S}

\newcommand{\ex}{\operatorname{ex}}

\newcommand{\bz}{{\bm z}}

\newcommand{\Span}{\operatorname{span}}

\newcommand{\gid}{\operatorname{gid}}

\newcommand{\bi}{\mathbf i}
\newcommand{\bj}{\mathbf j}

\newcommand{\Det}{\operatorname{Det}}

\newcommand{\change}[1]{{#1}}

\newtheorem{theorem}{Theorem}[section]

\newtheorem{lemma}[theorem]{Lemma}
\newtheorem{corollary}[theorem]{Corollary}
\newtheorem{proposition}[theorem]{Proposition}
\newtheorem{question}[theorem]{Question}
\newtheorem{conjecture}[theorem]{Conjecture}
\newtheorem{claim}[theorem]{Claim}
\newtheorem{definition}{Definition}[section]
\newtheorem{remark}{Remark}[section]
\newtheorem{example}{Example}[section]

\hypersetup{
	colorlinks=true,
	linkcolor=blue,%
	citecolor=blue,
	linktoc=page
}

\title{Generalized GM-MDS: \\Polynomial Codes are Higher Order MDS\footnote{A preliminary version of this manuscript appeared in the proceedings of STOC 2024~\cite{brakensiek2024Generalized}. This expanded versions gives full proofs of all claims.}}

\author{Joshua Brakensiek\thanks{University of California, Berkeley. Email: \texttt{josh.brakensiek@berkeley.edu}. Research supported in part by a Microsoft Research PhD fellowship while a student at Stanford University.}  \and
Manik Dhar\thanks{Department of Mathematics, Massachusetts Institute of Technology. Email: \texttt{dmanik@mit.edu}. Part of this work was done while this author was a graduate student at the Department of Computer Science, Princeton University where his research was supported by NSF grant DMS-1953807.} \and
Sivakanth Gopi\thanks{Microsoft Research. Email: \texttt{sigopi@microsoft.com}}}
\date{}

\begin{document}
\maketitle

\begin{abstract}
The GM-MDS theorem, conjectured by Dau-Song-Dong-Yuen and proved by Lovett and Yildiz-Hassibi, shows that the generator matrices of Reed-Solomon codes can attain every possible configuration of zeros for an $\MDS$ code. The recently emerging theory of higher order MDS codes has connected the GM-MDS theorem to other important properties of Reed-Solomon codes, including showing that Reed-Solomon codes can achieve list decoding capacity, even over fields of size linear in the message length. 

A few works have extended the GM-MDS theorem to other families of codes, including Gabidulin and skew polynomial codes. In this paper, we generalize all these previous results by showing that the GM-MDS theorem applies to any \emph{polynomial code}, i.e., a code where the columns of the generator matrix are obtained by evaluating linearly independent polynomials at different points. We also show that the GM-MDS theorem applies to dual codes of such polynomial codes, which is non-trivial since the dual of a polynomial code may not be a polynomial code. More generally, we show that GM-MDS theorem also holds for \emph{algebraic codes} (and their duals) where columns of the generator matrix are chosen to be points on some irreducible variety which is not contained in a hyperplane through the origin. Our generalization has applications to constructing capacity-achieving list-decodable codes as shown in a follow-up work \cite{bdg2023b}, where it is proved that randomly punctured algebraic-geometric (AG) codes achieve list-decoding capacity over constant-sized fields.
\end{abstract}
\newpage
\tableofcontents
\newpage

\section{Introduction}

The optimal rate distance tradeoff of an $(n,k)$-code (over large alphabets) is given by the Singleton bound~\cite{singleton1964maximum} which states that the minimum distance $d$ is at most $n-k+1$. A code which achieves this bound is called an MDS (Maximum Distance Separable) code. MDS codes are some of the most important codes because of their optimality, for example they are extensively used in distributed storage. Reed-Solomon codes~\cite{reed1960polynomial} are the most well-known family of MDS codes and one of the first polynomial codes to be discovered. An $(n,k)$ Reed-Solomon code is obtained by evaluating degree $<k$ polynomials over $\F_q$ at $n$ distinct points from $\F_q$. These codes achieve the Singleton bound over any field of size $q \ge n$. In this paper, we will only consider linear codes defined over finite fields. It is not hard to show that a linear $(n,k)$-code $C$ with generator matrix $G_{k\times n}$ is MDS iff for every $k\times k$ minor of $G$ is nonzero. The generator matrix of a Reed-Solomon code is of the form \begin{equation}
  \label{eq:rs_generator}
  G_{k\times n}=\left[\begin{matrix}
    1 & 1 & \dots & 1\\
    \alpha_1 & \alpha_2 & \dots & \alpha_n\\
    \alpha_1^2 & \alpha_2^2 & \dots & \alpha_n^2\\
    \vdots & \vdots & \ddots & \vdots\\
    \alpha_1^{k-1} & \alpha_2^{k-1} & \dots & \alpha_n^{k-1}
  \end{matrix}\right]
\end{equation}
where $\alpha_1,\alpha_2,\dots,\alpha_n\in \F_q$ are distinct. There are several natural generalizations of MDS codes and recently many of them have been unified into a single concept called \emph{higher order MDS codes} \cite{bgm2021mds,bgm2022}. Higher order MDS codes tie together the GM-MDS theorem, optimally list-decodable codes and MR tensor codes. We will now briefly discuss these connections.

\paragraph{GM-MDS theorem.} In some coding theory applications, we need MDS codes whose generator matrices have constrained support, i.e., with a specific pattern of zeros. See \cite{dau2014gmmds,Halbawi2014distributed,yan2013algorithms,dau2013balanced,dau2014simple} for applications of such codes to multiple access networks and secure data exchange. Let $\cS=(S_1,S_2,\dots,S_n)$ where $S_1,\dots,S_k\subset [n] \change{:= \{1, 2, \hdots, n\}}$, we call such an $\cS$ a zero pattern for $k\times n$ matrices. If there are $\ell$ distinct non-empty sets among $S_1,\dots,S_k$, we say that $\cS$ is an order-$\ell$ zero pattern. We say that an $(n,k)$-code $C$ with some generator matrix $G_{k\times n}$ attains the zero pattern $\cS$ if there exists an invertible matrix $M_{k\times k}$ such that $\tilde{G}=MG$ has zeros in the positions specified by $\cS$ (i.e., in positions $\cup_{i=1}^k \{i\}\times S_i$). Since $\tilde{G}$ is also a generator matrix of $C$, equivalently, we say $C$ attains the zero pattern $\cS$ if there exists a generator matrix of $C$ which has zeros in positions of $\cS$. For an MDS code to attain a zero pattern $\cS$, a simple necessary condition is that the zero pattern should satisfy a Hall type condition:
\begin{equation}
  \label{eq:hall_zp}
	\forall I\subset [k],\ |\cap_{i\in I} S_i|\le k-|I|.
\end{equation}
 It is easy to see why (\ref{eq:hall_zp}) is a necessary condition, if (\ref{eq:hall_zp}) is not satisfied, then there exists a $k\times k$ submatrix of $G$ which has too many zeros and hence cannot be invertible. On the other hand, if (\ref{eq:hall}) is satisfied, then by Hall's matching theorem every $k\times k$ submatrix has some matching of non-zeros entries and hence can be made invertible by a suitable (or generic) choice of the non-zero entries \cite{dau2014gmmds,dau2014simple}. Therefore, over large enough alphabet, one can always find an MDS matrix with a given zero pattern $\cS$ if (\ref{eq:hall}) is satisfied. A pattern which satisfies condition (\ref{eq:hall_zp}) is called a \emph{generic zero pattern}~\cite{bgm2022}.
 
 Dau et al.~\cite{dau2014gmmds,dau2014simple} conjectured that given any generic zero pattern, there exists a Reed-Solomon code over any field of size $q\ge n+k-1$ which attains the zero pattern $\cS$. This conjecture which came to be known as the \emph{GM-MDS conjecture}\footnote{It is unclear what GM stands as \cite{dau2014gmmds,dau2014simple} don't explain it in the paper where they introduced the conjecture. Our best guess is that it stands for `Generator Matrix' given that we are looking for MDS codes with constrained generator matrices.} was eventually proved independently by Lovett~\cite{lovett2018gmmds} and Yildiz and Hassibi~\cite{yildiz2019gmmds}. The heart of the proof (using a reduction shown by \cite{dau2014gmmds,dau2014simple}) is to show that generic Reed-Solomon codes can achieve the zero pattern $\cS$ which itself reduces to showing that a symbolic matrix has non-zero determinant. By observing that the degree of each variable in the determinant polynomial is at most $n+k-2$, one can immediately get an upper bound on the field size needed to achieve the zero pattern $\cS$ \cite{dau2014gmmds,dau2014simple}. So the GM-MDS theorem can equivalently be stated as follows.

\begin{theorem}[GM-MDS Theorem~\cite{dau2014gmmds,dau2014simple,lovett2018gmmds,yildiz2019gmmds}]
  \label{thm:gm-mds}
  Let $\cS=(S_1,S_2,\dots,S_n)$ be any generic zero pattern. Then for any finite field $\F_q$, there exists some $\alpha_1,\dots,\alpha_n\in \overline{\F}_q$ (algebraic closure of $\F_q$) such that the Reed-Solomon code with generator matrix $G$ as in (\ref{eq:rs_generator}) attains the zero pattern $\cS$.
\end{theorem}
An MDS code which achieves all generic zero patterns of order $\ell$ is called a $\GZP(\ell)$ code~\cite{bgm2022}. In other words, the GM-MDS theorem states that Reed-Solomon codes are $\GZP(\ell)$ for all $\ell$. GM-MDS type theorems have been proved for Gabidulin codes~\cite{gabidulin2021rank} in \cite{yildiz2019gabidulin} and for linearized Reed-Solomon codes (or skew polynomial codes) in \cite{liu2023linearized}. Are these codes somehow very special or is there a common generalization of these results?

\begin{question}
  \label{q:gm-mds}
  Is there a common generalization of the GM-MDS theorem which implies the existing results for Reed-Solomon codes, Gabidulin codes and linearized Reed-Solomon codes?
\end{question}

Note that all these codes are polynomial codes. We will see shortly that we can in fact prove a GM-MDS theorem for all polynomial codes!

\paragraph{Optimally list-decodable codes.} A code with minimum distance $d$ can be uniquely decoded up to half the minimum distance. \emph{List-decoding} introduced by \cite{wozencraft1958list,elias1957list} allows us to decode from beyond half the minimum distance, by instead outputting a small list of possible codewords. An $(n,k)$-code is $(\rho,L)$-list decodable if any Hamming ball of radius $\rho n$ contains at most $L$ codewords. \emph{Average-radius list-decoding} is a slight strengthening of list-decoding defined as follows.
\begin{definition}
  Let $\wt(x)$ denote the hamming weight (i.e., number of nonzeros) of a vector $x$. We say that an $(n,k)$-code is $(\rho,L)$-average-radius list-decodable if there doesn't exist $L+1$ codewords $c_0,c_1,\dots,c_L\in C$ and $y\in \F^n$ such that $$\frac{1}{L+1}\sum_{i=0}^{L}\wt(c_i-y) \le \rho n.$$
\end{definition}
A code which is $(\rho,L)$-average-radius list-decodable is also $(\rho,L)$-list-decodable. A generalized Singleton bound for list-decoding was proved recently~\cite{shangguan2020combinatorial,roth2021higher,goldberg2021singleton} which states that a $(\rho,L)$-list-decodable (and hence also average-radius list-decodable) code should satisfy 
\begin{equation}
  \label{eq:listdecoding_singleton}
  \rho \le \frac{L}{L+1}\change{\cdot \frac{n-k}{n}}.
\end{equation}
A code which meets this bound exactly for all list sizes $\le L$ is called an $\LDMDS(\le L)$ (list-decoding MDS) code \change{and is an example of a higher order MDS code~\cite{roth2021higher}}.

\paragraph{MR tensor codes.} Give an $(m,m-a)$-code $C_{col}$ and an $(n,n-b)$-code $C_{row}$, the tensor code $C_{col}\otimes C_{row}$ consists of all $m\times n$ matrices where each row belongs to the row code $C_{row}$ and each column belongs to the column code $C_{col}$. Tensor codes are used for erasure coding in data centers. \change{For example the f4~\cite{muralidhar2014f4} architecture has the parameters $(m,n,a,b) = (3,14,1,4)$.} A Maximally Recoverable (MR) tensor code are `optimal' tensor codes which can recover from any erasure pattern as long as it is information theoretically possible~\cite{Gopalan2016}. \cite{Gopalan2016} also give an explicit combinatorial condition on which erasure patterns are information theoretically recoverable which called `regularity'.

\paragraph{Higher order MDS codes.} \change{Motivated by higher order MDS codes,} Brakensiek, Gopi and Makam\change{ in \cite{bgm2021mds}} introduced a \change{different} generalization of MDS codes called \emph{higher order MDS codes}.
\begin{definition}[\cite{bgm2021mds}]\label{def:MDS-ell}
  Let $C$ be an $(n,k)$-code with generator matrix $G$. For $\ell \ge 1$, we say $C$ is $\MDS(\ell)$ if for any subsets $A_1,A_2,\dots,A_\ell\subset [n]$, we have that $$\dim(G_{A_1}\cap G_{A_2}\cap \dots \cap G_{A_\ell})=\dim(W_{A_1}\cap W_{A_2} \cap \dots W_{A_\ell})$$ where $G|_{A_i}$ is the submatrix of $G$ with columns indexed by $A_i$, $G_{A_i}$ is the column span of $G|_{A_i}$, and $W_{k\times n}$ is a generic matrix (over fields of the same characteristic).\footnote{Actually, the results of \cite{bgm2022} show that the RHS is independent of the characteristic of the field.}
\end{definition}

\change{As observed by \cite{bgm2021mds}, the conditions $\ell=1$ and $\ell=2$ coincide with ordinary higher order MDS codes, but the conditions for $\ell \ge 3$ are strictly stronger.} In their follow up work, \cite{bgm2022} showed a remarkable equivalence between various natural generalizations of MDS codes by showing that all of them are equivalent to higher order MDS codes. 

\begin{theorem}[\cite{bgm2022}]\label{thm:mds-equiv}
  The following are equivalent for any $(n,k)$-code $C$ and $\ell \ge 2$.
  \begin{enumerate}
    \item $C$ is $\MDS(\ell)$.
    \item $C$ is $\GZP(\ell)$.
    \item $C_{parity} \otimes C$ is an MR tensor code where $C_{parity}$ is the $(\ell,\ell-1)$ parity check code.
    \item $C^\perp$ is $\LDMDS(\le \ell-1)$ where $C^\perp$ is the dual code of $C$.
  \end{enumerate}
\end{theorem}

As a result of this theorem, an equivalent way to state the GM-MDS theorem is to state that generic Reed-Solomon codes are higher order MDS codes, i.e., $\MDS(\ell)$ for all $\ell$. In particular, this implies that Reed-Solomon codes over sufficiently large alphabet are optimally list decodable, i.e., they achieve the list-decoding singleton bound (\ref{eq:listdecoding_singleton}) for all list sizes, which was originally conjectured by Shangguan and Tamo \cite{shangguan2020combinatorial}. In subsequent work, Guo and Zhang \cite{guo2023randomly} showed that Reed-Solomon codes over fields of size $O_\eps(nk)$ can achieve list-decoding capacity, i.e., $(\rho,L)$ list decodable with $\rho=1-R-\eps$ and $L=O(1/\eps)$ where $R=k/n$ is the rate of the code. Note that this is weaker than being optimally list-decodable, but still very useful for applications and crucially now requires only $O_\eps(nk)$ field size. In fact this relaxation is necessary, \cite{brakensiek2023improved} showed that achieving $\LDMDS(\le 2)$ (i.e., optimally list decodable with for list sizes $L\le 2$) already requires exponential field sizes. Shortly after, Alrabiah, Guruswami and Li \cite{alrabiah2023randomly} further improved this to show that Reed-Solomon codes over $O_\eps(n)$ fields can achieve list-decoding capacity.

Given the rich set of connections of higher order MDS codes, answering Question~\ref{q:gm-mds} or equivalently understanding what codes have the higher order MDS property is a fundamental question which will likely lead to progress in several areas of coding theory such as list-decoding, MR tensor codes etc..

\subsection{Main Results}

We now discuss the main contributions of this paper. We satisfactorily answer Question~\ref{q:gm-mds} by showing that there is nothing very special about Reed-Solomon, Gabidulin or Linearized Reed-Solomon codes. We prove that GM-MDS theorem holds for all polynomial codes.

\begin{definition}[Polynomial codes]
\label{def:polycode}
Let $f=(f_1,f_2,\dots,f_k) \in (\F[x_1,\dots,x_r])^k$ be a $k$-dimensional vector of polynomials such that $f_1,\dots,f_k$ are linearly independent over $\F$. Let $\alpha_1,\alpha_2,\dots,\alpha_\change{n}\in \F^r$ be distinct. The polynomial code $C_f \change{\subseteq \F^n}$ is defined as the linear code with generator matrix $G_{k\times n}$ where 
\begin{equation}
	\label{eq:polycode_generator}
  G= \begin{bmatrix}
    f_1(\alpha_1) & f_1(\alpha_2) &  \cdots & f_1(\alpha_\change{n})\\
    f_2(\alpha_1) & f_2(\alpha_2) &  \cdots & f_2(\alpha_\change{n})\\
    \vdots & \vdots & \ddots & \vdots\\
    f_k(\alpha_1) & f_k(\alpha_2) &  \cdots & f_k(\alpha_\change{n})
  \end{bmatrix}.
\end{equation}
\change{If $\alpha_1,\dots,\alpha_\change{n}\in \overline{\F}^r$ are chosen to be algebraically independent over $\F$, we say that $C_f \change{\subseteq \overline{\F}^n}$ is generic over $\F$.}\end{definition}

\begin{theorem}[(Informal) Generalized GM-MDS Theorem for Polynomial Codes]\label{thm:main-1}
  Generic polynomial codes are higher order MDS, i.e., is $\MDS(\ell)$ for all $\ell$.
  \end{theorem}

\begin{example} 
  \label{ex:known-gm-mds}
  We can recover all the known examples of GM-MDS theorem using Theorem~\ref{thm:main-1}.
  \begin{enumerate}
    \item Reed-Solomon codes are polynomial codes with $f(x)=(1,x,x^2,\dots,x^{k-1})$.
    \item Gabidulin codes are polynomial codes with $f(x)=(x,x^q,x^{q^2},\dots,x^{q^{k-1}})$ where $|\F|=q^m$ for some $m$.
    \item Linearized Reed-Solomon~\change{\cite{liu2015Construction}} or Skew polynomial codes~\change{\cite{martinez-penas2018Skew}} are polynomial codes with $$f(x,y)=(x,x^qy,x^{q^2}y^{1+q},\dots,x^{q^{k-1}}y^{1+q+\dots+q^{k-2}})$$ where $|\F|=q^m$ for some $m$\change{, where we pick evaluation points $(\alpha_1, \beta_1), \hdots, (\alpha_n, \beta_n) \in \F^2$. The GM-MDS theorem in~\cite{liu2023linearized} allows for some choices of the $\alpha_i$'s to be equal, but the choices of the $\beta_i$'s are distinct, so apriori it is not clear whether Theorem~\ref{thm:main-1} implies such a GM-MDS result. However, let $\gamma_i := \beta_i / \alpha_i^{q-1}$ and observe that
    \[
        f(\alpha_i, \beta_i) = \alpha_i (1, \gamma_i, \gamma_i^{q+1}, \hdots, \gamma_i^{1+q+\cdots +q^{k-2}}).
    \]
    The choice of scalar $\alpha_i$ does not affect the GM-MDS property, so the GM-MDS in this case follows by applying Theorem~\ref{thm:main-1} to the curve $g(z) = (1, z, \hdots, z^{1+q+\cdots+q^{k-2}})$. In general, linearized Reed-Solomon codes can having shared values among both the $\alpha_i$'s and the $\beta_i$'s. However, our GM-MDS theorem does not apply in such cases. 
    }
  \end{enumerate}
\end{example}

If a polynomial code $f=(f_1,\hdots,f_k)$ is composed of degree at most $d$ polynomials then it is not hard to show that any particular order $\ell$ pattern can be achieved by $[n,k]$-code over a field of size $d(\ell-1)+1$ (this does not include the condition that the code is $\MDS$). \change{This is established by Lemma~\ref{lem:mds5} stated in Section~\ref{sec:prelim} to reduce this algebraic condition to verifying that a particular symbolic determinant is nonzero for some choice of evaluation points. More precisely, the individual degree of each variable is at most $d(\ell-1)$.} As mentioned earlier, for Reed-Solomon codes an $\MDS$ code which satisfies a particular zero pattern can be constructed over a field of size $q\ge n+k-1$ (this is because one could simplify the determinant in Lemma~\ref{lem:mds5} to get better individual degree bounds). In a follow up paper~\cite{bdg2023b}, starting from a generalization of these theorems \change{we show} how to construct codes which satisfy `relaxed' generic zero patterns over a field of constant size (see end of this section).

All known examples of codes satisfying GM-MDS theorem such as Reed-Solomon, Gabidulin, Linearized Reed-Solomon are self-duality property, i.e., their dual codes are also of the same kind. For example, the dual of a Reed-Solomon code is a again a (generalized) Reed-Solomon code. However, this is not true for most polynomial codes. And since a code $C$ being $\LDMDS(\le L)$ (i.e., optimally list decodable) is equivalent to the dual code $C^\perp$ being $\MDS(L+1)$, it is important to know whether duals of generic polynomials codes are also higher order MDS and hence optimally list-decodable.
\begin{theorem}[(Informal) Generalized GM-MDS Theorem for Dual Polynomial Codes]\label{thm:main-1-dual}
  Duals of generic polynomial codes are higher order MDS, i.e., is $\MDS(\ell)$ for all $\ell$. As a result, generic polynomial codes and their duals are $\MDS(\ell)$ for all $\ell$ and $\LDMDS(\le L)$ for all $L$.
  \end{theorem}

We can generalize Theorems~\ref{thm:main-1} and \ref{thm:main-1-dual} even further. Instead of specifying an explicit parameterization for $f$, we \change{specify that} the columns of our code \change{are} solutions to \change{a} set of polynomial equations (or, more formally an \emph{irreducible variety}).

\begin{theorem}[(Informal) Generalized GM-MDS for irreducible varieties]
  \label{thm:main-var}
  Let $X\subset \overline{\F}^k$ be an irreducible variety which is not contained in any hyperplane. Let $v_1,v_2,\dots,v_n$ be generic points from $X$. Then the $(n,k)$-code $C$ with generator matrix $G=[v_1,v_2,\dots,v_k]$ as well as the dual code $C^\perp$ are $\MDS(\ell)$ for all $\ell$.
\end{theorem}
Theorems~\ref{thm:main-1} and \ref{thm:main-1-dual} are \change{special cases} of this theorem where the variety is $X=\{f(x):x\in \overline{\F}^r\}$. Note that any variety which can be parametrized is irreducible and the linear \change{i}ndependence of $f_1,f_2,\dots,f_k$ implies that $X$ is not contained in any hyperplane.

In a follow up work~\cite{bdg2023b}, Theorem~\ref{thm:main-var} is crucially used to prove that randomly punctured Algebraic-Geometric codes achieve list decoding capacity over \change{constant-sized fields}.

\subsection{Technical Overview}

We now outline how each of our main results are proved. First, it is possible to reduce the GM-MDS theorem for arbitrary irreducible varieties to the case of monomial codes. A monomial code is a polynomial code as in Definition~\ref{def:polycode} where each $f_i$ is a monomial. We first describe an overview of the proof for monomial codes and then discuss the reduction.

\subsubsection{Generalized GM-MDS Theorem for Monomial Codes}

We prove Theorem~\ref{thm:main-1} by an inductive approach similar to how the original GM-MDS theorem was proved \cite{lovett2018gmmds,yildiz2019gmmds}. We first recall how the GM-MDS theorem is proved.

The work of \cite{dau2014gmmds,dau2014simple} showed that the GM-MDS theorem is equivalent to the following fact. Let $A_1, \hdots, A_k \subseteq [n]$ be sets of size $k-1$, such that for any $J \subseteq [k]$ nonempty, we have that $|\bigcap_{j \in J} A_j| \le k - |J|$. It suffices to show for any sufficiently large field $\F$, (i.e., $|\F| \ge n+k-1$), there exists distinct $\alpha_1, \hdots, \alpha_n \in \F$ such that the polynomials $g_i(x) := \prod_{j \in A_i} (x - \alpha_j)$ are linearly independent over $\F$.

The works \cite{lovett2018gmmds,yildiz2019gmmds} prove this linear independence by inductively proving a generalized statement where they allow each polynomial $g_i$ to have $0$ as a repeated root.\footnote{The natural question of allowing each $g_i$ to have multiple repeated roots is in general open, see \cite{greaves2019reed}. One important observation is that the answer to the problem depends on the characteristic of the field, so it cannot have as simple of an answer as the GM-MDS theorem.} In particular, for each $i \in [k]$, they keep track of a set $A_i \subseteq [n]$ and $\sigma_i \ge 0$ such that $\sigma_i + |A_i| = k-1$. They inductively show that the polynomials $g_{i}(x) := x^{\sigma_i} \prod_{j \in A_i} (x - \alpha_j)$ are linearly dependent over $\F$ if and only if for all $J \subseteq [k]$ nonempty, we have that \begin{align}
\min_{j \in J} \sigma_j + \left|\bigcap_{j \in J} A_j\right| \le k-|J|. \label{eq:hall-intro}
\end{align}
Outside of simple base cases, the inductive proof splits into three cases (per the convention used by \cite{yildiz2019gmmds}):

\paragraph{Case 1.} There is $J \subseteq [k]$ with $2 \le |J| \le k-1$ for which (\ref{eq:hall-intro}) is tight. Define $g_J(x) := x^{\min_{j \in J} \sigma_j} \prod_{i \in \bigcap_{j \in J} A_j} (x - \alpha_i)$. The authors use the inductive hypothesis twice: first they argue that $\Span \{g_i : i \in J\} = \Span \{x^i g_J(x) : i \in \{0, 1, \hdots, |J|-1\}\}$. They then apply a second use of the induction hypothesis is to argue that $\{x^i g_J(x) : i \in \{0, 1, \hdots, |J|-1\}\} \cup \{g_j(x) : j \in [k] \setminus J\}$ is linearly independent.

\paragraph{Case 2.} There is exactly one $i \in [\ell]$ for which $\sigma_i = 0$. Note that at least one $\sigma_i = 0$ or else (\ref{eq:hall-intro}) fails for $J = [k]$. Since every $g_j(x)$ with $j \neq i$ is a multiple of $x$, it suffices to argue that $\{g_j(x) / x : j \in [k] \setminus i\}$ are linearly independent over $\F$, which follows from the inductive hypothesis.

\paragraph{Case 3.} Any other situation. In this case the authors identify an index $j \in [n]$ such that after assigning $\alpha_j \to 0$--that is, if $j \in A_i$, remove $j$ from $A_i$ and increase $\sigma_i$ by $1$--we have that (\ref{eq:hall-intro}) still holds. After applying the inductive hypothesis, they show algebraically that undoing this operation preserves the linear independence of the $g_i$.

\paragraph{\change{From Polynomials to Vectors: Higher Order MDS Codes with a Basis.}}  \change{Our first main technical contribution is translate the techniques used by \cite{lovett2018gmmds,yildiz2019gmmds} from the language of polynomials to that of vectors. A portion of this translation appears in the \emph{higher order MDS} framework of \cite{bgm2021mds,roth2021higher,bgm2022}. In particular, Theorem~\ref{thm:mds-equiv} of \cite{bgm2022} shows that proving a variant of the GM-MDS theorem is equivalent to the following. Let $C$ be our polynomial code with generator matrix $G$ as in (\ref{eq:polycode_generator}). Consider any $A_1, \hdots, A_k \subseteq [n]$ of size $k-1$ such that for all $J \subseteq [k]$, we have that $|\bigcap_{j \in J} A_j| \le k - |J|$. We seek to show that $G_{A_1} \cap \cdots \cap G_{A_k} = 0$. In the case for which $G$ is the generator matrix of a Reed--Solomon code, the connection is that $G_{A_i}$ is the dual of the space of polynomials $f$ of degree at most $k$ which have roots at $\alpha_j$ for all $j \in A_i$.
}

\change{To generalize the GM-MDS theorem to Theorem~\ref{thm:main-1}, it suffices to understand the structure of $G_{A_1} \cap \cdots \cap G_{A_k}$, where $G$ is an arbitrary (generic) polynomial code. However, the current framework is insufficient for encoding an induction analogous to that of \cite{lovett2018gmmds,yildiz2019gmmds}. In particular, it is not clear how to encode polynomials which have $0$ as root with multiplicity $\sigma_i$ in this framework. We do this by introducing the concept of \emph{higher order MDS codes with a basis}. Here, we define a basis to be an ordered set of vectors $u_1, \hdots, u_k \in \F^k$. With this basis, we can interpret "a zero with multiplicity $\sigma_i$'' as adding to the vector space $G_{A_i}$ the first $\sigma_i$ vectors from $u_1, \hdots, u_k$. We define the concept more precisely as follows.}

\begin{definition}\label{def:MDSb}
  Consider $n,k,b \ge 1$ with $b \le k$. Let $U \in \F^{k \times b}$
  and $V \in \F^{k \times n}$ be matrices. \change{We call $U$ the \emph{basis} matrix (with the columns of $U$ the basis vectors) and $V$ the \emph{generator} matrix.} We say that $(U,V)$ is a
  $(b,n,k)$-\emph{order-$\ell$ higher order MDS code with a basis}, or
  $(b,n,k)$-$\MDSb(\ell)$, if \change{for all $\sigma_1, \hdots, \sigma_\ell \in \{0,1,\hdots, b\}$ and $A_1, \hdots, A_\ell \subseteq [n]$,} we have that
  \begin{align}
    \dim \left(\bigcap_{i=1}^\ell (U_{[\sigma_i]} + V_{A_i})\right) = \dim \left(\bigcap_{i=1}^\ell (W_{[\sigma_i]} + W'_{A_i})\right), \label{eq:mds-basis'}
  \end{align}
  where $W \in \F^{k \times b}$ and $W' \in \F^{k \times n}$ are generic matrices \change{algebraically independent of each other}.
\end{definition}

The generalized statement we prove is that that for any code $C$ with a generator matrix like (\ref{eq:polycode_generator}), there is a matrix $U \in \F^{k \times k}$ such that $(U, G)$ is \change{$(k,n,k)$-}$\MDSb(\ell)$ for all $\ell \ge 2$\change{. Note} by invoking (\ref{eq:mds-basis'}) in the case that each $\sigma_i = 0$, we get that $G$ is $\MDS(\ell)$.  Once a suitable basis is identified for our code $C$, the strategy of \cite{lovett2018gmmds,yildiz2019gmmds}, can be followed; although it takes some effort to show that conditions like (\ref{eq:mds-basis'}) behave analogous to results on repeated roots of polynomials--this groundwork is laid in Section~\ref{sec:mds-basis}.

In the case of Reed-Solomon codes, we can set $U$ to be the identity matrix; that is, the appropriate choice of basis is just the standard basis vectors $e_1, \hdots, e_k \in \F^k$. In general for polynomial codes, the choice of basis depends on the degrees of the nonzero coefficients of $f_1, \hdots, f_k$. However, for ``monomial codes'', where each $f_i(x) = x^{d_i}$, for some integers $d_1 < d_2 < \cdots < d_k$, we can again use the standard basis. In Section~\ref{sec:algebra} we give all the details of the proof for monomial codes, and in Section~\ref{sec:irr-var}, we give a reduction from polynomial codes to monomial codes (see also Section~\ref{subsub:irr-var}).

\subsubsection{Generalized GM-MDS Theorem for Duals of Monomial Codes} 

As explained earlier, except in special cases, a generic polynomial code $C$ does not necessarily have that the property that its dual code $C^{\perp}$ is a generic polynomial code. Therefore, we cannot use the (exact) same approach to prove that $C^{\perp}$ is $\MDS(\ell)$ (and thus $C$ is $\LDMDS(\le \ell-1))$ for any generic polynomial code $C$.

An approach which one might take is to directly compute a parity check matrix $H$ of $C$ in terms of the entries of the generator matrix $G$ from (\ref{eq:polycode_generator}). However, the entries of $H$ depend on the values of many of the $\alpha_i$'s and thus make it nearly impossible to directly show that $H$ is $\MDS(\ell)$, even with a suitable basis.

Instead, we show a surprising connection between higher order MDS codes with a basis and dual codes. In particular, we show that if $(G^{T}, I_n)$ is \change{$(k,n,n)$-}$\MDSb(\ell)$, then $C^{\perp}$ is $\MDS(\ell)$ (see Theorem~\ref{thm:gm-mdsb-dual}). In other words, if we think of the \emph{rows} of $G$ as a (partial) basis of $\F^n$, it suffices to show that the criteria of (\ref{eq:mds-basis'}) hold when the ``main'' part of the code is the standard basis of $\F^n$.

Like in the case of Theorem~\ref{thm:main-1}, we can assume  for the proof of Theorem~\ref{thm:main-1-dual} that the generator matrix $G$ of the code $C$ consists of monomials. In this case, we can prove that $(G^{T}, I_n)$ is \change{$(k,n,n)$-}$\MDSb(\ell)$ using a GM-MDS-style induction (see Theorem~\ref{thm:tr-gm-mds}). A crucial detail in the proof is that since the algebraic complexity of analyzing statements like (\ref{eq:mds-basis'}) comes from $G^{T}$, the inductive proof \emph{removes} ``repeated roots'' rather than adding them. In particular, the base case in which each $\sigma_i = 0$ is easy to analyze due to $I_n$ being the standard basis. Once the order of induction is properly chosen, the remainder of the proof of Theorem~\ref{thm:main-1-dual} is similar to that of Theorem~\ref{thm:main-1}.

\subsubsection{Extension to Polynomial Codes and Irreducible Varieties}\label{subsub:irr-var}

So far, we have only explained the proofs of Theorem~\ref{thm:main-1} and Theorem~\ref{thm:main-1-dual} for monomial codes, where each polynomial of (\ref{eq:polycode_generator}) has a single nonzero coefficient. To extend this result to general polynomials, we observe that the higher order MDS conditions are equivalent to certain determinants based on the columns of the generator matrix $G$ being nonzero (see Lemma~\ref{lem:mds5}). We show that once a suitable change of basis is performed to $G$, a leading monomial in the determinant expansion corresponds to the analogous determinant where each polynomial of $G$ is replaced by its leading monomial. This only works when the column of the generator matrix comes from the same polynomial family. As such, the GM-MDS theorems we proved for monomial codes also apply to general polynomial codes.

We also extend these results to general irreducible varieties. In this case, we use the fact that any irreducible variety can be (generically) parameterized by a \emph{power series}, essentially a polynomial where there are infinitely many terms. These power series can be analyzed essentially like polynomials, in particular they can be reduced to analyzing a suitable monomial code. As such, Theorem~\ref{thm:main-var} follows from the proofs of Theorem~\ref{thm:main-1} and Theorem~\ref{thm:main-1-dual}.

\change{We also note that one could directly adapt the GM-MDS theorems' proof for power series without reducing to a univariate monomial code as one only needs the property that power series are linearly independent. In this paper we present the proof of these theorems after the monomial reduction to keep the exposition of the generalized GM-MDS theorems simple so that the main ideas are clearly illustrated.}

\subsection*{Organization}

In Section~\ref{sec:prelim}, we define our notation and state some known facts about higher order MDS codes. In Section~\ref{sec:mds-basis}, we formally define a Higher Order MDS code with a basis and use this to build an inductive framework for proving GM-MDS-type theorems. In Section~\ref{sec:algebra}, we prove our GM-MDS theorem for monomial codes. In Section~\ref{sec:duality}, we prove our GM-MDS theorem for duals of monomial codes. In Section~\ref{sec:irr-var}, we show the results of Sections~\ref{sec:algebra} and~\ref{sec:duality} extend to polynomial codes and irreducible varieties.  \change{In Section~\ref{sec:openquestions}, we conclude the paper with some open questions.} In Appendix~\ref{app}, we include proofs omitted from the main paper.

\change{\subsection*{Acknowledgements}

We thank anonymous referees for numerous helpful comments which assisted in greatly improving the quality of this manuscript.}

\section{Preliminaries}\label{sec:prelim}

In this section, we go over some technical background material not explained in the introduction.

\subsection{Notation}

In this paper we let $\F$ denote an arbitrary field. Although applications to coding theory require $\F$ to be finite, our results generally apply to any field (of any characteristic). 

A \emph{linear code} is a subspace $C \subseteq \F^n$, where $n$ is number of symbols of the code. If $k$ is the dimension of $n$, then we say that $C$ is a $(n,k)$-code. We call points $x \in C$ \emph{codewords} of $C$. The \emph{Hamming weight} of a codeword $c \in C$ is the number of nonzero symbols which we denote by $\wt(c)$. A \emph{generator matrix} $G \in \F^{k \times n}$ of a linear code $C \subseteq \F^n$ is any matrix whose rows form a basis of $C$.

The \emph{dual code} of $C$ is the linear code $C^{\perp} := \{y \in \F^n \mid \forall x \in C, x_1y_1 + \cdots + x_ny_n = 0\}$. We say that a generator matrix of the dual code $C^{\perp}$ is a \emph{parity check matrix} of the original code $C$. 

We say that $C$ is \emph{maximum distance separable (MDS)} if a generator matrix $G \in \F^{k \times n}$ of $C$ has every $k$ columns linearly independent. This is equivalent to the more traditional definition of MDS that every nonzero codeword of $C$ has Hamming weight at least $n-k+1$.

Given a matrix $G \in \F^{k \times n}$, for $i \in [n] := \{1,2, \hdots, n\}$, we let $G_i$ denote the $i$th column of $G$. For a subset $A \subseteq [n]$, we let $G|_{A}$ denote the submatrix of $G$ with columns from $A$. We also let $G_A := \Span \{G_i \mid i \in A\}$. Note that some previous works like \cite{bgm2021mds} use $G|_{A}$ and $G_A$ interchangeably.  In this paper, we often use a code and its generator matrix interchangeable if the property of the matrix is independent of the choice of generator matrix. For example, the rank of $G|_{A}$ is such a property.

If we have subsets $S \subseteq [k]$ and $T \subseteq [n]$, we let $G|_{S \times T}$ denote the submatrix with entries $(i,j) \in S \times T$.

Given two matrices $U$ and $V$ with the same number of rows, we let $[U\mid V]$ denote the matrix with the columns of $U$ and $V$ concatenated. We let $I_k$ denote the $k \times k$ identity matrix. If $k$ is implicit from context, we write just $I$ (e.g., $[I \mid V]$).

Given a field $\F$ and parameters $(n,k)$, we define a \emph{generic $(n,k)$-matrix} as follows. Let $\F_{\ex}$ be an extension of $\F$ with $kn$ transcendental\footnote{That is, there is no nonzero polynomial equation in the $w_{i,j}$'s with coefficients in $\F$ which evaluates to $0$. In other words, the $w_{i,j}$'s can be thought of as ``symbolic'' variables.} variables $\{w_{i,j} : (i,j) \in [k] \times [n]\}$.  Our generic matrix is $W \in \F_{\ex}^{k \times n}$ where $W_{i,j} = w_{i,j}$ for all $(i,j) \in [k] \times [n]$. Given a polynomial code such as in (\ref{eq:polycode_generator}), we say that it is a \emph{generic polynomial code} if $\alpha_1, \hdots, \alpha_n$ are algebraically independent, transcendental variables in a suitable extension of $\F$.

\subsection{Properties of Higher Order MDS Codes}

We now state a few properties of higher order MDS codes that did not appear in the introduction. \change{First, we quote a result by \cite{bgm2022} on the intersections of vector spaces. More precisely, given a generic $k\times n$ matrix $W$ and sets $A_1, \hdots, A_\ell \subseteq [n]$, what is the dimension of $W_{A_1} \cap \cdots \cap W_{A_\ell}$? It turns out that there is an explicit combinatorial formula.}

\change{\begin{theorem}[Generic intersection dimension formula, Theorem~1.16~\cite{bgm2022}]\label{thm:gid}
Let $W \in \F^{k \times n}$ be a generic matrix. Let $A_1, \hdots, A_\ell \subseteq[n]$ be subsets of size at most $k$. We have that
\begin{align}
    \dim(W_{A_1} \cap \cdots \cap W_{A_\ell}) = \max_{P_1 \sqcup \cdots \sqcup P_s = [\ell]}\left[-(s-1)k + \sum_{i=1}^s |A_{P_i}|\right].\label{eq:gid}
\end{align}
\end{theorem}}

\change{To appreciate why (\ref{eq:gid}), note that if $V_1, \hdots, V_s \subseteq \F^k$ are vector spaces, then
\[
    k - \dim(V_1 \cap \cdots \cap V_s) \ge \sum_{i=1}^k (k - \dim V_i). 
\]
In particular, if we have a lower bound on the dimension of each $V_i$, we get a lower bound on the dimension of $V_1 \cap \cdots \cap V_s$. Here, we set each $V_i$ to be $\bigcap_{j \in P_i} W_{A_j}$ and note that $\dim(V_i) \ge |A_{P_i}|$, which immediately yields the ``$\ge$'' direction of (\ref{eq:gid}). The surprising contribution of Theorem~\ref{thm:gid} is that these simple bounds are tight.}

\change{However, the full identity (\ref{eq:gid}) is a bit cumbersome to use. Instead, we often only need to understand whether $W_{A_1} \cap \cdots \cap W_{A_\ell} = 0.$ Such an expression can be easily extracted from Theorem~\ref{thm:gid} by defining a combinatorial ``null intersection property.''}

\begin{definition}[\change{Inspired by \cite{bgm2021mds,bgm2022}}]
  We say that $A_1, \hdots, A_\ell \subseteq [n]$ \change{of size at most $k$} have the \emph{\change{(combinatorial)} dimension $k$ null intersection property} if for all partitions $P_1 \cup \cdots \cup P_s = [\ell]$, we have that
  \begin{align}
    \sum_{i=1}^s \left|\bigcap_{j \in P_i} A_j \right| \le (s-1)k.\label{eq:combo-null}
  \end{align}
\end{definition}

\change{With this definition, a simplified version of Theorem~\ref{thm:gid} easily follows.}

\begin{corollary}[\change{Implicit in \cite{bgm2022}}]\label{cor:mds5}
  Let $V$ be a $k \times n$ matrix. Then, $V$ is $(n,k)$-$\MDS(\ell)$
  if and only if for all $A_1, \hdots, A_\ell \subseteq [n]$ with
  $|A_i| \le k$ and $|A_1| + \cdots + |A_\ell| = (\ell-1)k$ with the dimension $k$ null intersection property, we have
  that $V_{A_1} \cap \cdots \cap V_{A_\ell} = 0$.
\end{corollary}

\begin{proof}
The right hand size of (\ref{eq:gid}) equals zero if and only if $-k(s-1) + \sum_{i=1}^s |A_{P_i}| \le 0$ for all partitions $P_1 \sqcup \cdots \sqcup P_s = [\ell]$. This is precisely (\ref{eq:combo-null}). Thus, this result follows immediately from Theorem~\ref{thm:gid}.
\end{proof}

\change{In some of our results, we need to utilize the technical machinery used to prove Theorem~\ref{thm:gid}, including the following ``duality'' result where the sizes of the intersections of the sets $A_1, \hdots, A_\ell$ can be governed by \emph{dual variables} $\delta_1, \hdots, \delta_\ell$ (analogous to those which arise in linear programs).}

\begin{theorem}[Lemma~2.8~\cite{bgm2022}, adapted]\label{thm:mds6}
  Let $A_1, \hdots, A_\ell \subseteq [n]$ be subsets of size at most $k$. Let $d \in \{0, 1, \hdots, k\}$, the following are equivalent
  \begin{itemize}
    \item[(a)] In \change{every} $\MDS(\ell)$ code $V$, we have that $\dim(V_{A_1} \cap \cdots \cap V_{A_\ell}) \le d$.
    \item[(b)] For all partitions $P_1 \sqcup \cdots \sqcup P_s = [\ell]$, we have that
      \[
        \sum_{i=1}^s \left|\bigcap_{j \in P_i} A_j\right| \le (s-1)k + d.
      \]
    \item[(c)] There exist $\delta_1, \hdots, \delta_\ell \ge 0$ with $\delta_1 + \cdots + \delta_\ell = k-d$ such that for all $J \subseteq [n]$ nonempty, we have that
      \[
        \left|\bigcap_{j \in J} A_j\right| \le k - \sum_{j \in J} \delta_j.
      \]
  \end{itemize}
\end{theorem}

\change{Since we need to argument about the algebraic properties of various intersections of vector spaces, we} also use the following matrix identity for checking \change{vector space intersection} conditions.

\begin{lemma}[\cite{tian2019formulas,bgm2021mds}]\label{lem:mds5}
  Let $V$ be a $k \times n$ matrix. Consider
  $A_1, \hdots, A_\ell \subseteq [n]$ with $|A_i| \le k$. We have that
  \begin{align}
    \dim(V_{A_1} \cap \cdots \cap V_{A_\ell}) &= k + \sum_{i=1}^{\ell} \dim(V_{A_i}) - \rank \begin{pmatrix}
    I_k & V|_{A_1} & & & \\
    I_k & & V|_{A_2} & &\\
    \vdots & & & \ddots &\\
    I_k & & & &  V|_{A_\ell}
  \end{pmatrix}.\label{eq:mat-inter-dim}
  \end{align}
  Furthermore, if  $|A_1| + \cdots + |A_\ell| = (\ell-1)k$ and $\dim(V_{A_i}) = |A_i|$ for all $i \in [\ell]$, then we have that $V_{A_1} \cap \cdots \cap V_{A_\ell} = 0$ if and only if the block matrix in (\ref{eq:mat-inter-dim}) has nonzero determinant.
\end{lemma}

\section{Higher Order MDS Codes with a Basis}\label{sec:mds-basis}

In this section, we generalize the definition of higher order MDS codes to \change{to require that some vectors can only be added to subspaces in a fixed sequence, which we call a ``basis.''} As mentioned in the technical overview, this will allow us to substantially generalize the ``repeated roots'' trick in the proof of the GM-MDS theorem.

\subsection{Definition}

Consider $n,k,b \ge 1$ with $b \le k$. Let $U \in \F^{k \times b}$ \change{be a matrix}. We define an \emph{order-$\ell$ configuration} to be $\ell$ pairs $(\sigma_1, A_1), \hdots, (\sigma_\ell,  A_\ell)$, where for each $i \in [\ell]$ we have that $0 \le \sigma_i \le b$ and $A_i \subseteq [n]$. For all $J \subseteq [\ell]$ nonempty, let
\begin{align*}
  \sigma_J &:= \min_{i \in J} \sigma_i\\
  A_J &:= \bigcap_{i \in J} A_i.
\end{align*}
Let $U_{\le \sigma} = \Span \{u_i : i \in [\sigma]\}$, and $U|_{\le \sigma}$ denote the submatrix of the first $\sigma$ columns of $U$.

\begin{definition}[\change{Definition~\ref{def:MDSb} restated}]\label{def:MDSb-alt}
  Consider $n,k,b \ge 1$ with $b \le k$. Let $U \in \F^{k \times b}$
  and $V \in \F^{k \times n}$ be matrices. \change{We call $U$ the \emph{basis} matrix (with the columns of $U$ the basis vectors) and $V$ the \emph{generator} matrix.} We say that $(U,V)$ is a
  $(b,n,k)$-\emph{order-$\ell$ higher order MDS code with a basis}, or
  $(b,n,k)$-$\MDSb(\ell)$, if for all \change{order-$\ell$} configurations
  $(\sigma_1, A_1), \hdots, (\sigma_{\ell}, A_\ell)$, we have that
  \begin{align}
    \dim \left(\bigcap_{i=1}^\ell (U_{\le \sigma_i} + V_{A_i})\right) = \dim \left(\bigcap_{i=1}^\ell (W_{\le \sigma_i} + W'_{A_i})\right), \label{eq:mds-basis}
  \end{align}
  where $W \in \F^{k \times b}$ and $W' \in \F^{k \times n}$ are generic matrices \change{algebraically independent of each other}.
\end{definition}

\begin{remark}
\change{We note that $U$ can also be defined without selecting a basis nor a generator matrix. In particular, in (\ref{eq:mds-basis}) one only needs to specify $0 = U_{\le 0} \subsetneq U_{\le 1} \subsetneq \cdots \subsetneq U_{\le b} \subseteq \F^k.$ This is known in a linear algebra as a \emph{flag}.}
\end{remark}

\subsection{Comparison to \texorpdfstring{$\MDS(\ell)$}{MDS(l)} Codes}

To help the reader get intuition for Definition~\ref{def:MDSb-alt}, we compare and contrast with Definition~\ref{def:MDS-ell}.

Like in \cite{bgm2021mds}, this definition is monotone in $\ell$. That is, if $(U,V)$ is $\MDSb(\ell)$, then it is $\MDSb(\ell')$ for all $\ell' \le \ell$. This can be observed by setting some of the pairs $(\sigma_i, A_i)$ in (\ref{eq:mds-basis}) equal to each other (c.f.~\cite{bgm2021mds}).

Another key observation is the following:

\begin{proposition}\label{prop:remove-basis}
If $(U,V)$ is $\MDSb(\ell)$, then $V$ is $\MDS(\ell)$. 
\end{proposition}

\begin{proof}
In (\ref{eq:mds-basis}), set $\sigma_i = 0$ for all $i \in [\ell]$, we then derive the condition in Definition~\ref{def:MDS-ell}.
\end{proof}

In particular, technique for proving that a code $V$ is $\MDS(\ell)$ is to find a suitable basis $U$ such that $(U,V)$ is $\MDSb(\ell)$. This stronger statement can lend itself to an inductive proof, see Section~\ref{subsec:induct}\change{.}

However, there are codes for which $(U,V)$ is $\MDSb(1)$ but $[U \mid V]$ is not MDS. For instance, let the ambient field be $\F_2$, have $U = I_3$, and $V$ be the matrix
\[
  V = \begin{pmatrix}
        1 & 1\\
        1 & 1\\
        0 & 1
      \end{pmatrix},
\]
then one can verify that $(U, V)$ is $\MDSb(1)$, but $[U \mid V]$ is not MDS, as the columns of $V$ and the third column vector of $U$ are linearly dependent.

As more nontrivial example, have $U = I_4$ and $V \in \F^{2 \times 4}$ be a generic Reed-Solomon code:
\[
    V = \begin{pmatrix}
        1 & 1\\
        \alpha & \beta\\
        \alpha^2 & \beta^2\\
          \alpha^3 & \beta^3
      \end{pmatrix}.
\] The \change{methods in the} original proofs of the GM-MDS theorem~\cite{lovett2018gmmds,yildiz2019gmmds} \change{can be used to} show that $(U,V)$ is $\MDSb(4)$. However, $[U \mid V]$ is not $\MDS(4)$. To see why, let $A_1 = \{1,2,5\}$, $A_2 = \{3,4,5\}$, $A_3 = \{1,3,6\}$, and $A_4 = \{2,4,6\}$, then one can show
\[
   [U \mid V]_{A_1} \cap [U \mid V]_{A_2} \cap [U \mid V]_{A_3} \cap [U \mid V]_{A_4}  \neq 0,
\]
as can be checked via Corollary~\ref{cor:mds5}, despite this identity being true for a generic code.

Although an $\MDSb(1)$ code $(U,V)$ is not necessarily MDS, every subset of the columns of $[U \mid V]$ which shows up in computing (\ref{eq:mds-basis}) \emph{does} behave like an MDS code. This is formalized in the following proposition.

\begin{proposition}\label{prop:mdsb-1}
Let $(U,V)$ be a $(b,n,k)$-$\MDSb(1)$ code. For any configuration $(\sigma, A)$, we have that
\[
  \dim(U_{\le \sigma} + V_A) = \min(\sigma + |A|, k).
\]
\end{proposition}

\begin{proof}
  For generic codes, $W$ and $W'$, we must have that $\dim(U_{\le \sigma} + V_A) = \dim(W_{\le \sigma} + W'_A) = \min(\sigma+|A|, k) $.
\end{proof}

As a corollary, we have that the $\ge$ direction of (\ref{eq:mds-basis}) follows whenever $(U,V)$ is $\MDSb(1)$ (c.f., \cite{bgm2021mds} or Lemma 3.1 of \cite{bgm2022}).

\begin{corollary}\label{cor:mdsb-ge}
Let $(U,V)$ is $\MDSb(1)$, then \change{for any $\ell \ge 1$ and} for all configurations $(\sigma_1, A_1), \hdots, (\sigma_\ell, A_\ell)$, we have that
\begin{align}
  \dim \left(\bigcap_{i=1}^\ell (U_{\le \sigma_i} + V_{A_i})\right) \ge \dim \left(\bigcap_{i=1}^\ell (W_{\le \sigma_i} + W'_{A_i})\right), \label{eq:mds-basis-ge}
\end{align}
where $W$ and $W'$ are generic matrices of the same shape as $U$ and $V$, respectively.
\end{corollary}

\begin{proof}
By Lemma~\ref{lem:mds5} and the fact that $\dim(U_{\le \sigma_i} + V_{A_i}) = \dim(W_{\le \sigma_i} + W'_{A_i})$ for all $i \in [\ell]$, we have that
\begin{align*}
&\dim \left(\bigcap_{i=1}^\ell (U_{\le \sigma_i} + V_{A_i})\right) -  \dim \left(\bigcap_{i=1}^\ell (W_{\le \sigma_i} + W'_{A_i})\right)\\
&= \rank \begin{pmatrix}
    I_k & W|_{\le \sigma_1} & W'|_{A_1} & & & & & \\
    I_k & & & W|_{\le \sigma_2} & W'|_{A_2} & & &\\
    \vdots & & & & \ddots & &\\
    I_k & & & & & W|_{\le \sigma_\ell} & W'|_{A_\ell}
  \end{pmatrix}\\&\hspace{0.3in} - \rank \begin{pmatrix}
    I_k & U|_{\le \sigma_1} & V|_{A_1} & & & & & \\
    I_k & & & U|_{\le \sigma_2} & V|_{A_2} & & &\\
    \vdots & & & & \ddots & &\\
    I_k & & & & & U|_{\le \sigma_\ell} & V|_{A_\ell}
  \end{pmatrix}.
\end{align*}
Since $(W,W')$ is generic, the substitution $(W,W') \to (U,V)$ can only decrease the rank\footnote{More formally, the condition that a matrix has rank at least $r$ is equivalent to whether there exists a $r \times r$ submatrix with nonzero determinant. Since the entries of $(W,W')$ are transcendental, any such determinant in the block matrix is necessarily nonzero for $(W,W')$ if it is nonzero for $(U,V)$.} of the block matrix. Thus, the final expression is nonnegative.
\end{proof}

\subsection{Properties of \texorpdfstring{$\MDSb(\ell)$}{MDSb(l)} Codes}

In this section, we shall prove various properties about $\MDSb(\ell)$ codes, which shall lead to an inductive framework for proving that various codes are $\MDSb(\ell)$, see Section~\ref{subsec:induct}.

We say that a configuration $(\sigma_1, A_1), \hdots, (\sigma_\ell, A_\ell)$ is \emph{proper} if $\sigma_i + |A_i| \le k$ for all $i \in [\ell]$. We observe that it suffices to check that a code if $\MDSb(\ell)$ by considering only proper configurations.

\begin{proposition}\label{prop:mdsb-proper}
If (\ref{eq:mds-basis}) \change{holds} for all proper configurations, then (\ref{eq:mds-basis}) holds for all configurations.
\end{proposition}

\begin{proof}
  We prove this by induction on $\ell$. The case $\ell=1$ follows from Proposition~\ref{prop:mdsb-1} and that $\dim(U_{\le \sigma} + V_A)$ is monotone in $(\sigma, A)$. For $\ell \ge 2$, note that if we have a configuration with some $\sigma_i + |A_i| > k$, then by the induction hypothesis, we have that $\dim(U_{\le {\sigma_i}} + V_{A_i}) = k$, so $U_{\le {\sigma_i}} + V_{A_i} = \F^k$. In particular, we can replace $(\sigma_i, A_i)$ with any $(\sigma'_i, A'_i)$ with $\sigma'_i + |A_i| = k$ and whether (\ref{eq:mds-basis}) holds will not change. We can repeat this procedure until we are left with a proper configuration.
\end{proof}

By adapting the generic intersection dimension formula of \cite{bgm2022}, we can get an explicit combinatorial formula for the RHS of (\ref{eq:mds-basis}) for proper configurations.

\begin{proposition}\label{prop:MDSb-inter-form}
Let $W \in \F^{k \times b}$ and $W' \in \F^{k \times n}$ be generic matrices \change{where the entries of $W$ and $W'$ are algebraically indepenent of each other}. Let $(\sigma_1, A_1), \hdots, (\sigma_\ell, A_\ell)$ be a proper configuration. Then,
\begin{align}
  \dim \left(\bigcap_{i=1}^\ell (W_{\le \sigma_i} + W'_{A_i})\right) = \max_{P_1 \sqcup \cdots \sqcup P_s = [\ell]} \left[-\change{k}(s-1) + \sum_{i=1}^s (\sigma_{P_i} + |A_{P_i}|)\right], \label{eq:MDSb-inter-form}
\end{align}
where $P_1 \sqcup \cdots \sqcup P_s = [\ell]$ denotes a partition of $[\ell]$ into $s$ nonempty pieces.
\end{proposition}

\begin{proof}
  Let $B := \{i_1, \hdots, i_b\}$ be some set of $b$ indices disjoint from $[n]$. For all $i \in [\ell]$, let 
\[
  C_i :=  A_i \cup \{i_1, \hdots, i_{\sigma_i}\}.
\]
It is easy to see that for all $J \subseteq [\ell]$ nonempty that $|C_J| = |\sigma_J| + |A_J|$. Thus, (\ref{eq:MDSb-inter-form}) is equivalent to
\[
  \dim \left(\bigcap_{i=1}^\ell [W|W']_{C_i}\right) = \max_{P_1 \sqcup \cdots \sqcup P_s = [\ell]} \left[-\change{k}(s-1) + \sum_{i=1}^s |C_{P_i}|\right],
\]
and this follows from \change{Theorem~\ref{thm:gid}.}
\end{proof}

For convenience, we let $\gid_k((\sigma_1, A_1), \hdots, (\sigma_\ell, A_\ell))$ denote the RHS of (\ref{eq:MDSb-inter-form}).

\begin{definition}\label{def:maximal}
We call a proper configuration $(\sigma_1, A_1), \hdots, (\sigma_\ell, A_\ell)$ to be \emph{maximal} if we additionally have the following properties
\begin{itemize}
\item[(a)] $\sum_{i=1}^\ell \sigma_i + |A_i| = (\ell-1)k.$
\item[(b)] For all partitions $P_1 \sqcup \cdots \sqcup P_s = [\ell]$, we have that
  \[
    \sum_{i=1}^s \sigma_{P_i} + |A_{P_i}| \le (s-1)k.
  \]
\item[(c)] There exists $\delta_1, \hdots, \delta_\ell \ge 0$ with $\sum_{i=1}^{\ell} \delta_i = k$ and for all $J \subseteq [n]$ nonempty, we have that
\[
  \sigma_J + |A_J| \le k - \sum_{i \in J} \delta_i.
\] 
\end{itemize}
\end{definition}

By adapting Theorem~\ref{thm:mds6}, we can prove that conditions (b) and (c) are equivalent.

\begin{proposition}\label{prop:dual-lp}
Let $(\sigma_1, A_1), \hdots, (\sigma_\ell, A_\ell)$ be \change{a proper configuration}. Then, conditions (b) and (c) of Definition~\ref{def:maximal} are equivalent.
\end{proposition}

\begin{proof}
  Let $b = \max_{i \in [\ell]} \sigma_i$. Let $B = \{i_1, \hdots, i_b\}$ be a set of size $\sigma$ disjoint from each of the $A_i$'s. As in the proof of Proposition~\ref{prop:MDSb-inter-form}, for all $i \in [\ell]$, define
  \[
    C_i = A_i \cup \{i_1, \hdots, i_{\sigma_i}\}.
  \]
  It is easy to see that for all $J \subseteq [\ell]$ nonempty that $|C_J| = |\sigma_J| + |A_J|$. Thus it suffices to show that the following two statements are equivalent:
  \begin{enumerate}
  \item[(b)] For all partitions $P_1 \sqcup \cdots \sqcup P_s = [\ell]$, we have that
    \[
      \sum_{i=1}^s |C_{P_i}| \le (s-1)k.
    \]
  \item[(c)] There exists $\delta_1, \hdots, \delta_\ell \ge 0$ with $\sum_{i=1}^{\ell} \delta_i = k$ and for all $J \subseteq [n]$ nonempty, we have that
    \[
      |C_J| \le k - \sum_{i \in J} \delta_i.
    \] 
  \end{enumerate}
  This equivalence follows immediately from the (b) $\iff$ (c) equivalence of Theorem~\ref{thm:mds6} with $d=0$.
\end{proof}

By adapting a result of \cite{bgm2021mds} (c.f., their Lemma B.5), we can show that any proper configuration can be ``mutated'' into a maximal configuration. Given $(\sigma_1, A_1), \hdots, (\sigma_\ell, A_\ell)$, we say that $(\sigma'_1, A'_1), \hdots, (\sigma'_\ell, A'_\ell)$ is a \emph{contraction} \change{of $(\sigma_1, A_1), \hdots, (\sigma_\ell, A_\ell)$} if $\change{\sigma'_i \le \sigma_i}$ and $A'_i \subseteq A_i$ for all $i \in [\ell]$. We say this contraction is \emph{minimal} if $\sum_{i=1}^{\ell} (\sigma_i - \sigma'_i) + (|A_i| - |A'_i|) = 1$. Likewise, we say that $(\sigma'_1, A'_1), \hdots, (\sigma'_\ell, A'_\ell)$ is an \emph{expansion} \change{of $(\sigma_1, A_1), \hdots, (\sigma_\ell, A_\ell)$} if $\change{\sigma'_i \ge \sigma_i}$ and $A'_i \supseteq A_i$ for all $i \in [\ell]$. We say this \change{expansion} is \emph{minimal} if $\sum_{i=1}^{\ell} (\sigma'_i - \sigma_i) + (|A'_i| - |A_i|) = 1$.

\begin{lemma}\label{lem:pad-mdsb}
Let $(\sigma_1, A_1), \hdots, (\sigma_\ell, A_\ell)$ be a proper configuration with $d = \gid_k((\sigma_1, A_1), \hdots, (\sigma_\ell, A_\ell))$. Then,
\begin{itemize}
\item[(a)] If $d \ge 1$, then there is a minimal contraction $(\sigma'_1, A'_1), \hdots, (\sigma'_\ell, A'_\ell)$ \change{of $(\sigma_1, A_1), \hdots, (\sigma_\ell, A_\ell)$} with \[d-1 = \gid_k((\sigma'_1, A'_1), \hdots, (\sigma'_\ell, A'_\ell)).\]

\item[(b)] If $d = 0$ and $\sum_{i=1}^\ell \sigma_i + |A_i| < (\ell-1)k$, then there is a minimal expansion $(\sigma'_1, A'_1), \hdots, (\sigma'_\ell, A'_\ell)$ \change{of $(\sigma_1, A_1), \hdots, (\sigma_\ell, A_\ell)$} with \[0 = \gid_k((\sigma'_1, A'_1), \hdots, (\sigma'_\ell, A'_\ell)).\]
\end{itemize}
\end{lemma}

As the details of this proof do not deviate substantially from those of \cite{bgm2021mds}, we leave the proof in the Appendix~\ref{app}. We thus, have the following generalization of Proposition~\ref{prop:mdsb-proper}.

\begin{lemma}\label{lem:case-1}
Let $(U,V)$ be a $(n,k,b)$-$\MDSb(1)$ code. The following are equivalent.

\begin{itemize}
\item[(a)] $(U,V)$ is $\MDSb(\ell)$.
\item[(b)] For all maximal configurations $(\sigma_1, A_1), \hdots, (\sigma_{\ell}, A_\ell)$, we have that
  \begin{align}
    \dim \left(\bigcap_{i=1}^\ell (U_{\le \sigma_i} + V_{A_i})\right) = 0.\label{eq:inter-basis}
  \end{align}
\end{itemize}
\end{lemma}
\begin{proof}
\change{We first prove that (a) implies (b).} \change{First, by Proposition~\ref{prop:MDSb-inter-form} and our assumption that $(U,V)$ is $\MDSb(\ell)$, we have for any maximal configuration $(\sigma_1, A_1), \hdots, (\sigma_{\ell}, A_\ell)$ that 
\[
\dim \left(\bigcap_{i=1}^\ell (U_{\le \sigma_i} + V_{A_i})\right) =   \dim \left(\bigcap_{i=1}^\ell (W_{\le \sigma_i} + W'_{A_i})\right) = \max_{P_1 \sqcup \cdots \sqcup P_s = [\ell]} \left[-\change{k}(s-1) + \sum_{i=1}^s (\sigma_{P_i} + |A_{P_i}|).\right]
\]
Now, by condition (b) of Definition~\ref{def:maximal}, we have that $-\change{k}(s-1) + \sum_{i=1}^s (\sigma_{P_i} + |A_{P_i}|) \le 0$ for every partition $P_1 \sqcup \cdots \sqcup P_s = [\ell]$. This proves the ``$\le$'' direction of (\ref{eq:inter-basis}). The ``$\ge$'' direction follows from the fact that dimension is nonnegative.}

We now focus on showing that (b) implies (a). Assume for sake of contradiction there exists a proper configuration $(\sigma_1, A_1), \hdots, (\sigma_{\ell}, A_\ell)$ such that (\ref{eq:mds-basis}) is violated. Since $(U,V)$ is $\MDSb(1)$, by Corollary~\ref{cor:mdsb-ge}, there are two possibilities
  \begin{itemize}
    \item[1.] $\gid_k((\sigma_1, A_1), \hdots, (\sigma_{\ell}, A_\ell)) = d\ge 1$ but $\dim \left(\bigcap_{i=1}^\ell (U_{\le \sigma_i} + V_{A_i})\right) \ge d+1$.
    \item[2.] $\gid_k((\sigma_1, A_1), \hdots, (\sigma_{\ell}, A_\ell)) = 0$, but $\dim \left(\bigcap_{i=1}^\ell (U_{\le \sigma_i} + V_{A_i})\right) \ge 1$.
  \end{itemize}

  In Case 1, we apply part (a) of Lemma~\ref{lem:pad-mdsb} to find a minimal contraction $(\sigma'_1, A'_1), \hdots, (\sigma'_{\ell}, A'_\ell)$ with $\gid_k((\sigma'_1, A'_1), \hdots, (\sigma'_{\ell}, A'_\ell)) = d-1$. Since the contraction is minimal, by Proposition~\ref{prop:sub1}, we have that
  \[
    \dim \left(\bigcap_{i=1}^\ell (U_{\le \sigma'_i} + V_{A'_i})\right) \ge \dim \left(\bigcap_{i=1}^\ell (U_{\le \sigma_i} + V_{A_i})\right) - 1 \ge d,
  \]
  where $Y = (U_{\le \sigma_i} + V_{A_i})$ and $Y' = U_{\le \sigma'_i} + V_{A'_i}$ for the pair which changed and $X$ is the intersection of the remaining spaces. By repeating this argument $d$ times, we eventually reach a proper configuration $(\sigma''_1, A''_1), \hdots, (\sigma''_{\ell}, A''_\ell)$ for which $\gid_k((\sigma''_1, A''_1), \hdots, (\sigma''_{\ell}, A''_\ell)) = 0$ but $\dim \left(\bigcap_{i=1}^\ell (U_{\le \sigma''_i} + V_{A''_i})\right) \ge 1$. Thus, we may now assume we are in Case 2.

 For Case 2, by Proposition~\ref{prop:MDSb-inter-form}, we must have that $\sum_{i=1}^\ell \change{\sigma}_i + |\change{A}_i| \le (\ell-1)k$. However, if we have equality, then we violate our assumption (b). Thus, we may assume the inequality is strict. We then apply part (b) of Lemma~\ref{lem:pad-mdsb} to find a minimal expansion $(\sigma'_1, A'_1), \hdots, (\sigma'_{\ell}, A'_\ell)$ \change{of $(\sigma_1, A_1), \hdots, (\sigma_\ell, A_\ell)$} with $\gid_k((\sigma'_1, A'_1), \hdots, (\sigma'_{\ell}, A'_\ell)) = 0$. Since we performed an expansion,
  \[
    \dim \left(\bigcap_{i=1}^\ell (U_{\le \sigma'_i} + V_{A'_i})\right) \ge \dim \left(\bigcap_{i=1}^\ell (U_{\le \sigma_i} + V_{A_i})\right) \ge 1.
  \]
  Thus, we still lie in Case 2. We repeat until we reach a maximal configuration $(\sigma''_1, A''_1), \hdots, (\sigma''_{\ell}, A''_\ell)$ for which $\dim \left(\bigcap_{i=1}^\ell (U_{\le \sigma''_i} + V_{A''_i})\right) \ge 1$, a contradiction.
\end{proof}

\subsection{Inductive Framework}\label{subsec:induct}

Assume we are seeking to prove that $(U,V)$ is $\MDSb(\ell)$ for $\ell \ge 2$, and assume we have already established that $(U,V)$ is $\MDSb(1)$. By Lemma~\ref{lem:case-1} and Proposition~\ref{prop:dual-lp}, it suffices to consider $(\ell,k)$-configurations tuples $(\sigma_1, A_1, \delta_1), \hdots, (\sigma_\ell, A_\ell, \delta_\ell)$ which satisfy the following conditions: 
\begin{itemize}
\item[(a*)] For all $i \in [\ell]$, $\sigma_i, \delta_i \ge 0$
\item[(b*)] $\sum_{i=1}^{\ell} \sigma_i + |A_i| = (\ell-1)k$
\item[(c*)] $\sum_{i=1}^{\ell} \delta_i = k$
\item[(d*)] For all $J \subseteq [\ell]$ nonempty, 
\begin{align}
\sigma_J + \left|A_J\right| \le k - \sum_{i \in J} \delta_i \label{eq:hall}
\end{align}
\end{itemize}
\change{We call such $(\sigma_1, A_1, \delta_1), \hdots, (\sigma_\ell, A_\ell, \delta_\ell)$ an \emph{$(\ell, k)$-null configuration}.} Note that conditions (b*), (c*), and (d*) together imply a fifth condition:
\begin{itemize}
\item[(e*)] For all $i \in [\ell]$, that $\sigma_i + |A_i| + \delta_i = k$.
\end{itemize}

Note by (\ref{eq:hall}) with $J = [\ell]$, we have that $\sigma_i = 0$ for at least one $i$, and for all $j \in [n]$ there is at least one $i \in [\ell]$ such that $j \not\in A_i$. For each \change{$(\ell,k)$-null configuration}, we seek to show that
\begin{align}
  \dim \left(\bigcap_{i=1}^\ell (U_{\le \sigma_i} + V_{A_i})\right) = 0.\label{eq:goal}
\end{align}

By Lemma~\ref{lem:mds5}, we have that (\ref{eq:goal}) is equivalent to
\begin{align}
  \det \begin{pmatrix}
    I_k & U|_{\le \sigma_1} & V|_{A_1} & & & & & \\
    I_k & & & U|_{\le \sigma_2} & V|_{A_2} & & &\\
    \vdots & & & & & \ddots & &\\
    I_k & & & & & & U|_{\le \sigma_\ell} &  V|_{A_\ell}
  \end{pmatrix} \neq 0\label{eq:goal-det}
\end{align}

As appears in the original proofs~\cite{lovett2018gmmds,yildiz2019gmmds} of the GM-MDS Theorem, an important special case is when there exists $J \subseteq [\ell]$ with $2 \le |J| \le \ell-1$ for which (\ref{eq:hall}) is tight. It turns out these cases follow only from the fact that $(U,V)$ is $\MDSb(\ell-1)$.

\begin{lemma}\label{lem:case-1-full}
  Let $(U,V)$ be $\MDSb(\ell-1)$ for $\ell \ge 2$. Let $(\sigma_1, A_1, \delta_1), \hdots, (\sigma_\ell, A_\ell, \delta_\ell)$ be \change{an $(\ell,k)$-null configuration}. Further assume there exists $J \subseteq [\ell]$ with $2 \le |J| \le \ell-1$ for which (\ref{eq:hall}) is tight. Then, (\ref{eq:goal}) holds.
\end{lemma}

\begin{proof}
Let $J^*$ be a specific set of size between $2$ and $n-1$ for which (\ref{eq:hall}) is tight. Like in Case 1 of the proof of the GM-MDS theorem, we can verify that for generic matrices $W \in \F^{k \times n}$ and $W' \in \F^{n \times n}$, we have that
\begin{align}
  \bigcap_{i \in J^*} (W_{\le \sigma_i} + W'_{A_i}) &= W_{\le \sigma_{J^*}} + W'_{A_{J^*}}\label{eq:part-a}\\
(W_{\le \sigma_{J^*}} + W'_{A_{J^*}}) \cap \bigcap_{i \not\in J^*} (W_{\le \sigma_i} + W'_{A_i}) &= 0.\label{eq:part-b}
\end{align}
In particular, to check (\ref{eq:part-a}), the $\supseteq$ inclusion is straightforward, so we just need to check both sides have the same dimension. For any partition $P_1 \sqcup \cdots \sqcup P_s = J^*$ we have that by (\ref{eq:hall}) that
\begin{align*}
  -\change{k}(s-1) + \sum_{i=1}^s (\sigma_{P_i} + |A_{P_i}|) &\le -k(s-1) + \sum_{i=1}^s \left[\change{k} - \sum_{j \in P_i} \delta_j\right]\\
                                                    &= k - \sum_{j \in J^{*}} \delta_j\\
                                                    &= |\sigma_{J^*}| + |A_{J^*}|, &\text{ (assumption on $J^*$)}
\end{align*}
where equality is achieved if $s=1$.

To check (\ref{eq:part-b}), it suffices to check (\ref{eq:hall}) (and invoke Proposition~\ref{prop:dual-lp}) for the tuple of sets $\{(\sigma_i, A_i, \delta_i) : i \not\in J^*\} \cup \{(\sigma_{J^*}, A_{J^*}, \delta_{J^*})\}$, where $\delta_{J^*} := \sum_{i \in J^*} \delta_i$. The intersections that do not include $(\sigma_{J^*}, A_{J^*}, \delta_{J^*})$ follow by assumption on the structure of the original configuration. Otherwise, for any $J' \subseteq J^*$ (possibly empty), we need to check that
\[
   \sigma_{J' \cup J^*} + |A_{J' \cup J^*}| \le k - \sum_{i \in J'} \delta_i + \delta_{J^*} = k - \sum_{i \in J' \cup J^*} \delta_i,
\]
which again follows from assumption.

To conclude this case, by the induction hypothesis, we have that $(U, V)$ is $\MDSb(\ell')$ for all $\ell' < \ell$. Thus, we have that both (\ref{eq:part-a}) and (\ref{eq:part-b}) hold with $(W,W')$ replaced by $(U, V)$. Note that for (\ref{eq:part-a}) we are using the fact that the $\supseteq$ direction holds unconditionally.  Therefore, (\ref{eq:goal}) holds.
\end{proof}

We use this inductive template in both Section~\ref{sec:algebra} and Section~\ref{sec:duality} to prove various generalizations of GM-MDS. \change{Each inductive proofs requires a different ordering on the set of configurations in order for the proof to go through, so we carefully describe each induction in the respective section.}

\section{Algebraic GM-MDS for Monomial Codes}\label{sec:algebra}

The goal of this section is to prove the following result. In Section~\ref{sec:irr-var}, we shall show how to extend this result to irreducible varieties: Theorem~\ref{thm-mdsIrred}.

\begin{theorem}\label{thm:gm-mds-primal}
Let $e_1 < e_2 < e_3 < \cdots < e_k$ be integers. Let $\alpha_1, \hdots, \alpha_n$ be generic points in an algebraically closed field $\F$. Then, for all $\ell \ge 2$, the \change{following} matrix \change{is $\MDS(\ell).$}
\begin{align}
  V := \begin{pmatrix}
    \alpha_1^{e_1} & \alpha_2^{e_1} & \cdots & \alpha_n^{e_1}\\
    \alpha_1^{e_2} & \alpha_2^{e_2} & \cdots & \alpha_n^{e_2}\\
         \vdots & \vdots & \ddots & \vdots\\
    \alpha_1^{e_k} & \alpha_2^{e_k} & \cdots & \alpha_n^{e_k}\\
  \end{pmatrix}\change{.} \label{mat:pseudo-rs}
\end{align}
\end{theorem}

By Proposition~\ref{prop:remove-basis}, it suffices to show that $(I\change{_k}, V)$ is \change{$(k,n,k)$-}$\MDSb(\ell)$ for all $\ell \ge 1$. We proceed by the inductive approach of Section~\ref{subsec:induct}. 

Consider an \change{$(\ell, k)$-null configuration $(\sigma_1, A_1, \delta_1), \hdots, (\sigma_\ell, A_\ell, \delta_\ell)$.} We seek to show (\ref{eq:goal}) that
\[
  \dim \left(\bigcap_{i=1}^\ell (\change{I}_{\le \sigma_i} + \change{V}_{A_i})\right) = 0.
\]

We shall prove this by induction on $\ell$, $k$ and $\sum_{i=1}^{\ell} |A_i|$. We first consider a few base cases.

\paragraph{Base cases.} For the base case of $\ell = 1$, it suffices to prove that for all $\sigma \in \{0, 1, \hdots, k\}$, we have that
\[
\det \begin{pmatrix}
1 & 0 & \cdots & \alpha_1^{e_1} & \alpha_2^{e_1} & \cdots & \alpha_{k-\sigma}^{e_1}\\
0 & 1  &  \cdots &  \alpha_1^{e_2} & \alpha_2^{e_2} & \cdots & \alpha_{k-\sigma}^{e_2}\\
 \vdots & \vdots & \ddots   &    \vdots & \vdots & \ddots & \vdots\\
0 & 0 &  \cdots & \alpha_1^{e_k} & \alpha_2^{e_k} & \cdots & \alpha_{k-\sigma}^{e_k}\\
\end{pmatrix} \neq 0,
\]
where the first $\sigma$ columns are standard basis vectors. This follows from the fact that the coefficient of $\alpha_1^{e_{\sigma + 1}} \cdots \alpha_{k-\sigma}^{e_{k}}$ in the determinant expansion is $\pm 1$.

The base case of $k = 1$ is obvious, as the generator matrix of any $(n,1)$-code with every entry nonzero is $\MDS(1)$ and thus $\MDS(\ell)$ for all $\ell$ (c.f., \cite{bgm2021mds}).

The final base case is $\sum_{i=1}^{\ell} |A_i| = 0$, or each $A_i = \emptyset$. In this case,
\[
\dim \left(\bigcap_{i=1}^\ell \change{I}_{\le \sigma_i}\right) = \dim \change{I}_{\le \sigma_{[\change{\ell}]}} = \sigma_{[\change{\ell}]} \change{= \min_{i \in [\ell]} \sigma_i},
\]
as desired.

\paragraph{Inductive step.} \change{Fix $\ell \ge 2, k \ge 2,$ and $N \ge 1$. As our inductive hypothesis we assume that (\ref{eq:goal}) holds for any $(\ell',k')$-null configuration $(\sigma'_1, A'_1, \delta'_1), \hdots, (\sigma'_{\ell'}, A'_{\ell'}, \delta'_{\ell'})$ whenever one of the following holds.
\begin{itemize}
\item $\ell' < \ell$ (in particular, we may assume that $(I_k, V)$ is $\MDS_b(\ell-1)$),
\item $\ell' = \ell$ and $k' < k$, or
\item $\ell' = \ell, k' = k'$ and $\sum_{i=1}^{\ell} |A'_i| < N$
\end{itemize}
We now fix an $(\ell, k)$-null configuration $(\sigma_1, A_1, \delta_1), \hdots, (\sigma_\ell, A_\ell, \delta_\ell)$ and seek to show it satisfies (\ref{eq:goal}). We divide this analysis into three cases.}

\begin{enumerate}

\item There is $J \subseteq [\ell]$ with $2 \le |J| \le \ell-1$ for which (\ref{eq:hall}) is tight.

\item $\sigma_i = 0$ for exactly one $i \in [\ell]$.

\item None of the above.

\end{enumerate}

\paragraph{Case 1.} This case follows immediately from Lemma~\ref{lem:case-1-full}.

\paragraph{Case 2.} Without loss of generality we may assume that $\sigma_1, \hdots, \sigma_{\ell-1} \ge 1$ and $\sigma_\ell = 0$. Recall we seek to verify that
\[
\det \begin{pmatrix}
    I_k & I_k|_{\le \sigma_1} & V|_{A_1} & & & & & & \\
    I_k & & & I_k|_{\le \sigma_2} & V|_{A_2} & & & & \\
    \vdots & & & & & \ddots & & &\\
    I_k & & & & & & I_k|_{\le \sigma_{\ell-1}} &  V|_{A_{\ell-1}} & \\
    I_k & & & & & & & &  V|_{A_\ell}
  \end{pmatrix} \neq 0
\]
Let $V'$ be $V$ with the first row removed. \change{We partially compute the above determinant by picking a judicious choice of columns. For each $i \in [\ell-1]$, consider the column block containing $I_{k}|_{\le \sigma_i}$ and $V|_{A_i}].$ Since $\sigma_i \ge 1$, we have that the first column of $I_{k}|_{\le \sigma_i}$ has exactly one nonzero entry. By expanding along this column, we effectively delete the first row of $I_{k}|_{\le \sigma_i}$ and $V|_{A_i}$, resulting in the block becoming $[I_{k-1}|_{\le \sigma_i-1}, V'|_{A_i}].$}

\change{With these $\ell-1$ column evaluations, the column block containing $\ell$ copies of $I_k$ now consists of $\ell-1$ copies of $I_k$ with the first row removed along with one untouched column of $I_k$. Thus, if we expand along the first column of the matrix, we delete the first row $V|_{A_\ell}$. After all these operations, we have established it is equivalent to check whether}
\begin{align}
\det \begin{pmatrix}
    I_{k-1} & I_{k-1}|_{\le \sigma_1-1} & V'|_{A_1} & & & & & & \\
    I_{k-1} & & & I_{k-1}|_{\le \sigma_2-1} & V'|_{A_2} & & & & \\
    \vdots & & & & & \ddots & & &\\
    I_{k-1} & & & & & & I_{k-1}|_{\le \sigma_{\ell-1}-1} &  V'|_{A_{\ell-1}} & \\
    I_{k-1} & & & & & & & &  V'|_{A_\ell}
  \end{pmatrix} \neq 0. \label{mat:step-2}
\end{align}

Consider $(\sigma'_1, A'_1, \delta'_1), \hdots, (\sigma'_\ell, A'_\ell, \delta'_\ell)$, where for $i \in [\change{\ell-1}]$, we have that $(\sigma'_i, A'_i, \delta'_i) = (\sigma_i -1 , A_i, \delta_i)$ and $(\sigma'_\ell, A'_\ell, \delta'_\ell) = (0, A_\ell, \delta_\ell-1)$. It is straightforward to check that: (a*), (b*), (c*) hold in $k-1$ dimensional space. For (d*), note that for any $J \subseteq [\ell-1]$ nonempty, we have that
\[
  \sigma'_J + \left|A'_J\right| = \sigma_J-1 + |A_J| \le k - 1 - \sum_{i \in J} \delta_i = k - 1 - \sum_{i \in J} \delta'_i.
\]
Thus, \change{$(\sigma'_1, A'_1, \delta'_1), \hdots, (\sigma'_\ell, A'_\ell, \delta'_\ell)$ is an $(\ell,k-1)$-null configuration.} \change{By} the induction hypothesis, we can deduce that (\ref{eq:goal-det}) holds for $(I_{k-1}, V')$, which is precisely (\ref{mat:step-2}).

\paragraph{Case 3.} We may assume that $\sigma_{\ell-1} = \sigma_\ell = 0$. By (e*) of the inductive framework, we know then that $|A_{\ell-1}| = k - \delta_{\ell-1}$ and $|A_{\ell}| = k - \delta_\ell$. Further by (d*), we have that $|A_{\ell-1} \cap A_{\ell}| \le k - \delta_\ell - \delta_{\ell-1}$. In particular, we must have that $\delta_{\ell-1} = \delta_\ell = 0$ or $A_{\ell-1} \neq A_\ell$. In the former case, it is straightforward to check that $(\sigma_1, A_1, \change{\delta_1}), \hdots, (\sigma_{\ell-1}, A_{\ell-1}, \delta_{\ell-1})$ satisfy conditions (a*)-(d*), and so by the induction hypothesis
\[
\bigcap_{i=1}^\ell (I_{\le \sigma_i} + V_{A_i}) = \bigcap_{i=1}^{\ell-1} (I_{\le \sigma_i} + V_{A_i}) = 0.
\]
In the latter case, we may without loss of generality assume there is $j^{*} \in A_{\ell-1} \setminus A_\ell$. \change{We now} define $(\sigma'_1, A'_1, \delta'_1), \hdots, (\sigma'_\ell, A'_\ell, \delta'_\ell)$ as follows. For all $i \in [\ell]$, let
\[
  (\sigma'_i, A'_i, \delta'_i) = \begin{cases}
                                   (\sigma_i, A_i, \delta_i) & j^{*} \not\in A_i\\
                                   (\sigma_i + 1, A_i \setminus \{j^{*}\}, \delta_i) & \text{otherwise}. 
                                 \end{cases}
\]
It is clear this new system satisfies conditions (a*)-(c*) of the inductive framework. Further, (d*) holds if $|J| = 1$ or $|J| = \ell$. Now assume $2 \le |J| \le \ell - 1$, observe that
\[
  \sigma'_J + |A'_J| \le \sigma_J + |A_J| + 1 \le k - \sum_{i \in J} \delta_i = k - \sum_{i \in J} \delta'_i,
\]
where the middle inequality follows from us not being in Case 1. Thus, \change{$(\sigma'_1, A'_1, \delta'_1), \hdots, (\sigma'_\ell, A'_\ell, \delta'_\ell)$ is an $(\ell,k)$-null configuration.} Note that $\sum_{i=1}^{\ell} |A'_i| < \sum_{i=1}^{\ell} |A_i|\change{= N}$. Therefore, by the induction hypothesis, we have that
\begin{align}
\bigcap_{i=1}^\ell (I_k|_{\le \sigma'_i} + V_{A'_i}) &= 0\change{.} \label{eq:gm-mds-mono}
\end{align}
To finish, it suffices to prove the following proposition\change{.}
\begin{proposition}
\[
\bigcap_{i=1}^\ell (I_k|_{\le \sigma'_i} + V_{A'_i}) = 0 \implies \bigcap_{i=1}^\ell (I_k|_{\le \sigma_i} + V_{A_i}) = 0\change{.}
\]

\end{proposition}
\begin{proof}
By Lemma~\ref{lem:mds5}, (\ref{eq:gm-mds-mono}) is equivalent to.
\begin{align}
  \det \begin{pmatrix}
    I_k & I_k|_{\le \sigma'_1} & V|_{A'_1} & & & & & & \\
    I_k & & & I_k|_{\le \sigma'_2} & V|_{A'_2} & & & & \\
    \vdots & & & & & \ddots & & &\\
    I_k & & & & & & I_k|_{\le \sigma'_{\ell-1}} &  V|_{A'_{\ell-1}} & \\
    I_k & & & & & & & &  V|_{A'_\ell}
  \end{pmatrix} \neq 0\label{mat:stuff}
\end{align}

Consider the exponent \[e = \sum_{\substack{i \in [\ell]\\j^* \in A_i}} e_{\sigma_i + 1}.\] Observe that the LHS of (\ref{mat:stuff}) is, up to a factor of $\pm 1$, the coefficient of $\alpha_{j^*}^e$ in 
\[
  \det \begin{pmatrix}
    I_k & I_k|_{\le \sigma_1} & V|_{A_1} & & & & & & \\
    I_k & & & I_k|_{\le \sigma_2} & V|_{A_2} & & & & \\
    \vdots & & & & & \ddots & & &\\
    I_k & & & & & & I_k|_{\le \sigma_{\ell-1}} &  V|_{A_{\ell-1}} & \\
    I_k & & & & & & & &  V|_{A_\ell}
  \end{pmatrix},
\]
which by Lemma~\ref{lem:mds5} implies that 
\[
\bigcap_{i=1}^\ell (I_k|_{\le \sigma_i} + V_{A_i}) = 0\change{.}\qedhere
\]
\end{proof}

This completes the proof of Theorem~\ref{thm:gm-mds-primal}.

\section{A GM-MDS Theorem for Dual Monomial Codes}\label{sec:duality}

The goal of this section is to prove the following result, which serves as an important ingredient for our more general theorem on $\LDMDS(\ell)$ codes (Theorem~\ref{thm-ldmdsIrred}). \change{Given a matrix $V \in \F^{k \times n}$ with rank $k$, we say that $Q \in \F^{(n-k) \times n}$ of rank $n-k$ is a dual matrix of of $V$ if for all $i \in [k]$ and $j \in [n-k]$, we have that $V_{i,1}Q_{j,1} + \cdots + V_{i,n}Q_{j,n} = 0$. If $V$ has rank $k$ it may have many dual matrices. We let $V^{\perp}$ denote the set of all such dual matrices. We remark that all dual matrices generate the same code (the dual to the code generated by $V$), so any properties concerning dual matrices we study are invariant of the choice of representative (cf. \cite{bgm2021mds,bgm2022}).}

\begin{theorem}\label{thm:gm-mds-dual}
Let $e_1 < e_2 < e_3 < \cdots < e_k$ be integers. Let $\alpha_1, \hdots, \alpha_n$ be generic points in an algebraically closed field $\F$. Then, for all $L \ge 2$, the matrix
\begin{align}
  V := \begin{pmatrix}
    \alpha_1^{e_1} & \alpha_2^{e_1} & \cdots & \alpha_n^{e_1}\\
    \alpha_1^{e_2} & \alpha_2^{e_2} & \cdots & \alpha_n^{e_2}\\
         \vdots & \vdots & \ddots & \vdots\\
    \alpha_1^{e_k} & \alpha_2^{e_k} & \cdots & \alpha_n^{e_k}\\
  \end{pmatrix} \label{mat:pseudo-rs-}
\end{align}
is $\LDMDS(\le L)$. In other words, \change{every dual matrix $Q \in V^{\perp}$} is $(n,n-k)$-$\MDS(L+1)$.
\end{theorem}

Note that the special case where each $e_i = i-1$ corresponds to Reed-Solomon codes. This result was established in \cite{bgm2022} by combining the GM-MDS theorem (Theorem~\ref{q:gm-mds}), Theorem~\ref{thm:mds-equiv}, and the fact that the dual of a Reed-Solomon code is a Reed-Solomon code up to column scaling (c.f., Proposition~4.4 of \cite{bgm2022} or \cite{macwilliams1977theory,hall7notes}). One would hope that we could use the same proof strategy, except replacing the GM-MDS theorem with Theorem~\ref{thm-mdsIrred}. However, the last step that ``the dual of a Reed-Solomon code is a Reed-Solomon code up to column scaling'' is false for this more general family of codes. Hence, we need to find an alternative proof of the GM-MDS theorem that works with this more general family of dual codes.

\subsection{Reduction}

Our plan for proving Theorem~\ref{thm:gm-mds-dual} is to (non-obviously) reduce the statement to the following fact.

\begin{theorem}\label{thm:tr-gm-mds}
Let $V$ be as in (\ref{mat:pseudo-rs}), then $(V^{T}, I_n)$ is $\MDSb(\ell)$ for all $\ell \ge 1$.
\end{theorem}

Rather unintuitively, the basis is the rows of $V^{T}$ and the ``normal'' part is $I_n$. We shall explain how this structure arises shortly, but first we shall prove the following fact.

\begin{proposition}\label{prop:mds-dual}
Let $V$ be an $(n,k)$-MDS matrix. Then, \change{every $Q \in V^{\perp}$} is $\MDS(\ell)$ if and only if for any $A_1, \hdots, A_\ell \subseteq [n]$ with the dimension $n-k$ null intersection property, $\sum_{i=1}^{\ell}|A_i| = (\ell-1)(n-k)$, and each set has size at most $n-k$, we have that 
\begin{align}
  \begin{pmatrix}
    I|_{\oA_1} & I|_{\oA_2} & \cdots & I|_{\oA_\ell}\\
    V|_{\oA_1} & & & \\
     & V|_{\oA_2} & & \\
    & &\ddots & \\
     & & & V|_{\oA_\ell}\\
  \end{pmatrix}\label{mat:rank}
\end{align}
has rank $n+(\ell-1)k$, where $\oA_i = [n] \setminus A_i$.
\end{proposition}

Note that (\ref{mat:rank}) has $n+(\ell-1)k$ rows but $n + \ell k$ columns.

\begin{proof}
Let $H \in \F^{(\ell-1) \times \ell}$ be a parity check matrix. By Theorem~\ref{thm:mds-equiv}, we have that $\change{Q \in V^{\perp}}$ is $\MDS(\ell)$ if and only if $H \otimes \change{Q}$ is an $(\ell,n,1,k)$-MR tensor code. More precisely, we have that $V_{A_1} \cap \cdots \cap V_{A_\ell} = 0$ if and only if the erasure pattern $E = \bigcup_{i=1}^{\ell} \{i\} \times \oA_i$ is recoverable for $H \otimes \change{Q}$. Note that since we assume that $\sum_{i=1}^{\ell}|A_i| = (\ell-1)(n-k)$, we have that

\[
 |E| = \sum_{i=1}^{\ell} |\oA_{i}| = n\ell - (\ell-1)(n-k) = n + (\ell-1)k.
\]
Observe that the all-ones row vector $1^{T}$ is \change{a} parity check matrix of $H$ and $V$ is \change{a} parity-check matrix of $\change{Q}$. Consider the matrix 
\begin{align}
  \begin{pmatrix}
    1^{T} \otimes I_n\\
    I_\ell \otimes V
  \end{pmatrix} = \begin{pmatrix}
                    I_n & I_n & \cdots & I_n\\
                    V & & & \\
                     & V & & \\
                    & & \ddots &\\
                    & & & V
                  \end{pmatrix}\label{mat:parity-check}
\end{align}
Note that every row of (\ref{mat:parity-check}) is a parity-check condition for the code $H \otimes \change{Q}$. Further the rows of (\ref{mat:parity-check}) span the parity check conditions of $H \otimes \change{Q}$. Thus, the pattern $E$ is recoverable if and only if the columns of (\ref{mat:parity-check}) corresponding to $E$ are linearly independent (c.f., \cite{bgm2021mds}). Observe that (\ref{mat:parity-check}) restrict to the columns indexed by $E$ is precisely (\ref{mat:rank}). Thus, (\ref{mat:rank}) has rank $|E| = n + (\ell-1)k$ if and only if $V_{A_1} \cap \cdots \cap V_{A_\ell} = 0$, as desired.
\end{proof}

We now show how Theorem~\ref{thm:tr-gm-mds} and Proposition~\ref{prop:mds-dual} complete the proof of Theorem~\ref{thm:gm-mds-dual}. This follows from the following theorem, and that fact that $V$ is $\MDS$, as proved in the base case of the proof of Theorem~\ref{thm:gm-mds-primal}.

\begin{theorem}\label{thm:gm-mdsb-dual}
  Let $V$ be an $(n,k)$-MDS matrix such that $(V^{T}, I_n)$ is $\MDSb(\ell)$. Then, \change{every $Q \in V^{\perp}$} is $\MDS(\ell)$. That is, $V$ is $\LDMDS(\le \ell-1)$.
\end{theorem}

\begin{proof}[Proof of Theorem~\ref{thm:gm-mdsb-dual}.]
Consider $A_1, \hdots, A_\ell \subseteq [n]$ with the dimension $n-k$ null intersection property, $\sum_{i=1}^{\ell}|A_i| = (\ell-1)(n-k)$, and each set has size at most $n-k$. We seek to show that (\ref{mat:rank}) has rank $n+(\ell-1)k$. Consider the proper order-$\ell$ configuration $(\sigma_1, A_1), \hdots, (\sigma_\ell, A_\ell)$ where each $\sigma_i = k$. Observe that for any partition $P_1 \sqcup \cdots \sqcup P_s = [\ell]$, we have that
\begin{align*}
-n(s-1) + \sum_{i=1}^s (\sigma_{P_i} + |A_{P_i}|) &= -n(s-1) + ks + \sum_{i=1}^s |A_{P_i}|\\ &= k + \left[ - (n-k)(s-1) +\sum_{i=1}^s |A_{P_i}|\right]\\ &\ge k,
\end{align*}
where equality for some partition as $A_1, \hdots, A_\ell$ have the dimension $n-k$ null intersection property.

Therefore, by Proposition~\ref{prop:MDSb-inter-form} and Theorem~\ref{thm:tr-gm-mds}, we have that
\[
 \dim \left(\bigcap_{i=1}^\ell (V^T_{\le k} + I_{A_i})\right) = k.
\] 

Equivalently, by Lemma~\ref{lem:mds5}, we have that the matrix

\begin{align}\begin{pmatrix}
I_n & V^{T} & I|_{A_1}&  & & & &\\
I_n &  & & V^{T} & I|_{A_2} & & & \\
\vdots & & & & & \ddots & &\\
I_n &  & &  & & & V^{T} & I|_{A_\ell}\\
\end{pmatrix}\label{mat:start}\end{align}
has rank \begin{align*}
-k + n+ \sum_{i=1}^{\ell}\dim(V^{T} + I_{A_1}) &= (n-k) + (\ell-1)(n-k)+\ell k\\ &= n\ell,
\end{align*}
where for the first line we use Proposition~\ref{prop:mdsb-1} and the fact that $(V^{T}, I)$ is $\MDSb(1)$.

Each $I|_{A_i}$ lets us eliminate $|A_i|$ rows. Doing this and permutating some rows/columns, we get that (\ref{mat:start}) has the same rank as
\begin{align}
\mleft[
\begin{array}{ccccc|ccc}
I|_{\oA_1 \times [n]} & V^{T}|_{\oA_1 \times [k]} & & & & & &\\
I|_{\oA_2 \times [n]} & & V^{T}|_{\oA_2 \times [k]} & & & & &\\
\vdots & & & \ddots & & & & \\
I|_{\oA_\ell \times [n]} & & & & V^{T}|_{\oA_\ell \times [k]} & & &\\\hline
& & & & & I|_{A_1 \times A_1} & &\\
& & & & &  & \ddots &\\
& & & & &  & & I|_{A_\ell \times A_\ell}
\end{array}
\mright]\label{mat:end}
\end{align}

Since the bottom-right submatrix of (\ref{mat:end}) is an identity matrix, we have that the top-left submatrix of (\ref{mat:end}) has rank $n\ell - \sum_{i=1}^\ell |A_i| = n\ell - (n-k)(\ell-1) = n + k(\ell-1)$. Since the top-left submatrix of (\ref{mat:end}) is the transpose of (\ref{mat:rank}), we have that \change{every $Q \in V^{\perp}$} is $\MDS(\ell)$ by Proposition~\ref{prop:mds-dual}, where we use that $V$ is MDS.
\end{proof}

\subsection{Proof of Theorem~\ref{thm:tr-gm-mds}}

In this section, we prove that $(V^{T}, I_n)$ is \change{$(k,n,n)$-}$\MDSb(\ell)$. We proceed by the inductive approach of Section~\ref{subsec:induct}. However there is a key difference between this proof and the others present so far. Consider a \change{$(k,\ell)$-null} configuration $(\sigma_i, A_i, \delta_i)$ which occurs in the induction. The GM-MDS theorem (and Theorem~\ref{thm-mdsIrred}) seeks to ``shrink'' the $A_i$'s and grow the corresponding $\sigma_i$'s. In our case, we will shrink the $\sigma_i$'s and grow the $A_i$'s. In the end, when $\sigma_i = 0$ for all $i \in [\ell]$, we can directly compute the intersection dimension, as $I_n$ is just the standard basis and thus is $\MDS(\ell)$. 

Consider a \change{$(k,\ell)$-null configuration $(\sigma_1, A_1, \delta_1), \hdots, (\sigma_\ell, A_\ell, \delta_\ell)$.} We seek to show (\ref{eq:goal}) that
\[
  \dim \left(\bigcap_{i=1}^\ell (V^T_{\le \sigma_i} + I_{A_i})\right) = 0.
\]

We shall prove this by induction on $\ell$ and $\sum_{i=1}^{\ell} \sigma_i$. We first consider a \change{two} base cases.

\paragraph{Base cases.} For the base \change{case} of $\ell=1$, we prove the stronger statement that $[V^{T} \mid I_n]$ is MDS.  Pick a subset of $n$ columns of this matrix, and expand the determinant along the columns corresponding to $I_n$. Up to sign, we are left with an expression of the general form
\[
\det \begin{pmatrix}
    \alpha_1^{e_1} & \alpha_1^{e_2} & \cdots & \alpha_1^{e_\sigma}\\
    \alpha_2^{e_1} & \alpha_2^{e_2} & \cdots & \alpha_2^{e_\sigma}\\
         \vdots & \vdots & \ddots & \vdots\\
    \alpha_\sigma^{e_1} & \alpha_\sigma^{e_2} & \cdots & \alpha_\sigma^{e_\sigma}\\
  \end{pmatrix} \neq 0,
\]
which follows from the proof of Theorem~\ref{thm:gm-mds-primal}.

For the other base case of $\sum_{i=1}^{\ell} \sigma_i = 0$, we have that $\sigma_i = 0$ for all $i$. In that case, we have that

\begin{align*}
   \dim \left(\bigcap_{i=1}^\ell (V^T_{\le \sigma_i} + I_{A_i})\right) &= \dim \left(\bigcap_{i=1}^\ell I_{A_i}\right)\\
   &= \dim I_{A_{[\ell]}} = |A_{[\ell]}|\\
   &= \max_{P_1 \sqcup \cdots \sqcup P_s = [\ell]} \left[-n(s-1) + \sum_{i=1}^s (\sigma_{P_i} + |A_{P_i}|)\right],
\end{align*}
where the last step uses the partition $P_1 = [\ell]$, as desired.

\paragraph{Inductive step.} \change{Fix $\ell \ge 2$ and $N \ge 1$. As our inductive hypothesis we assume that (\ref{eq:goal}) holds for any $(\ell',k')$-null configuration $(\sigma'_1, A'_1, \delta'_1), \hdots, (\sigma'_{\ell'}, A'_{\ell'}, \delta'_{\ell'})$ whenever one of the following holds.
\begin{itemize}
\item $\ell' < \ell$ (in particular, we may assume that $(V^T, I_n)$ is $\MDS_b(\ell-1)$),
\item $\ell' = \ell$ and $\sum_{i=1}^{\ell} \sigma_i < N$
\end{itemize}
Note that we do not induct on $k'$. We now fix an $(\ell, k)$-null configuration $(\sigma_1, A_1, \delta_1), \hdots, (\sigma_\ell, A_\ell, \delta_\ell)$ and seek to show it satisfies (\ref{eq:goal}). We divide this analysis into three cases.}

\begin{enumerate}

\item There is $J$ with $2 \le |J| \le \ell-1$ for which (\ref{eq:hall}) is tight.

\item $\sigma_i = \sigma_{i'} = 0$ for some $i \neq i'$

\item None of the above.

\end{enumerate}

\paragraph{Case 1.} This case follows immediately from Lemma~\ref{lem:case-1-full}.

\paragraph{Case 2.} Without loss of generality assume that $\sigma_{\ell-1} = \sigma_\ell = 0$. Note that 
\begin{align}
\bigcap_{i=1}^\ell (V^T_{\le \sigma_i} + I_{A_i}) = \bigcap_{i=1}^{\ell-2} (V^T_{\le \sigma_i} + I_{A_i}) \cap I_{A_{\ell-1} \cap A_{\ell}}.\label{eq:case-2}
\end{align}
The RHS corresponds to the \change{\emph{nearly} $(\ell-1,k)$-null} configuration $(\sigma_1, A_1, \delta_1), \hdots, (\sigma_{\ell-2}, A_2, \delta_{\ell-2}), (0, A_{\ell-1} \cap A_{\ell}, \delta_{\ell-1} + \delta_i).$ \change{By ``nearly'' we mean that the configuration satisfies conditions (a*), (c*), (d*) of the inductive framework, but not necessarily (b*).} Thus, by Proposition~\ref{prop:dual-lp} and the induction hypothesis that $(V^T, I_n)$ is $\MDSb(\ell-1)$, we have that the RHS of (\ref{eq:case-2}) equals $0$, as desired.

\paragraph{Case 3.} Without loss of generality, we may assume that $\sigma_1, \hdots, \sigma_{\ell-1} \ge 1$ but $\sigma_\ell = 0$. From (\ref{eq:hall}) with $J = [\ell-1]$, we have that
\begin{align*}
  1 + |A_{[\ell-1]}| &\le \sigma_{[\ell-1]} + |A_{[\ell-1]}|\\
                     &\le n - \sum_{i=1}^{\ell-1} \delta_i\\
  &= \delta_\ell\\
  &= n - |A_\ell|.
\end{align*}
Thus, there exists
$j^{*} \in [n] \setminus (A_\ell \cup A_{[\ell-1]})$. For all
$i \in [\ell]$, define $(\sigma'_i, A'_i, \delta'_i)$ as follows
\[
(\sigma'_i, A'_i, \delta'_i) =
\begin{cases}
(\sigma_i, A_i, \delta_i) & j^{*} \in A_i\text{ or } i = \ell\\
(\sigma_i - 1, A_i \cup \{j^{*}\}, \delta_i) & \text{ otherwise}.
\end{cases}
\]
We seek to show that $\{(\sigma'_i, A'_i, \delta'_i) : i \in [\ell]\}$ \change{is a $(\ell,k)$-null configuration}. Clearly conditions (a*)--(c*) are satisfied. To check (d*), we split into cases based on the size of $J \subseteq [\ell]$. If $|J| = 1$, then (\ref{eq:hall}) is still satisfied as $|\sigma'_i| + |A'_i| = |\sigma_i| + |A_i|$ for all $i \in [\ell]$. If $|J| = \ell$, then the LHS of (\ref{eq:hall}) is still $0$ as $\sigma_\ell = 0$ and $j^{*} \not\in A_\ell$. Otherwise, we have that
\begin{align*}
  \sigma'_J + \left|A'_J\right| &\le \sigma_J + \left|A_J\right| + 1& \text{ (only $j^{*}$ possibly added)}\\
                                &\le  n - \sum_{i \in J} \delta_i & \text{ (not in Case 1)}\\
                                &= n - \sum_{i \in J} \delta'_i.
\end{align*}

Since $j^{*} \not\in A_{[\ell-1]}$, we must have that $\sum_{i=1}^{\ell} \sigma'_i < \sum_{i=1}^{\ell} \sigma_i \change{=N}$. Therefore, by the induction hypothesis we have that
\[
\bigcap_{i=1}^\ell (V^T_{\le \sigma'_i} + I_{A'_i}) = 0.
\]

To complete the case, we need to show algebraically this implies (\ref{eq:goal}).

\begin{proposition}
We have that
\begin{align}
\bigcap_{i=1}^\ell (V^T_{\le \sigma'_i} + I_{A'_i}) = 0 \implies \bigcap_{i=1}^\ell (V^T_{\le \sigma_i} + I_{A_i}) = 0.\label{eq:goal'}
\end{align}
\end{proposition}

\begin{proof}
By Lemma~\ref{lem:mds5}, we have that the RHS of (\ref{eq:goal'}) is equivalent to
\begin{align}
\det \begin{pmatrix}
I_n & V^{T}|_{\le \sigma_1} & I|_{A_1}&  & & & &\\
I_n &  & & V^{T}|_{\le \sigma_2} & I|_{A_2} & & & \\
\vdots & & & & & \ddots & &\\
I_n &  & &  & & & V^{T}|_{\le \sigma_\ell} & I|_{A_\ell}
\end{pmatrix} \neq 0.\label{mat:RHS}
\end{align}
and the LHS of (\ref{eq:goal'}) is equivalent to
\begin{align}
\det\begin{pmatrix}
I_n & V^{T}|_{\le \sigma'_1} & I|_{A'_1}&  & & & &\\
I_n &  & & V^{T}|_{\le \sigma'_2} & I|_{A'_2} & & & \\
\vdots & & & & & \ddots & &\\
I_n &  & &  & & & V^{T}|_{\le \sigma'_\ell} & I|_{A'_\ell}
\end{pmatrix} \neq 0.\label{mat:LHS}
\end{align}

We reminder the reader that $V$ has the structure
\[
  V = \begin{pmatrix}
    \alpha_1^{e_1} & \alpha_2^{e_1} & \cdots & \alpha_n^{e_1}\\
    \alpha_1^{e_2} & \alpha_2^{e_2} & \cdots & \alpha_n^{e_2}\\
    \vdots & \vdots & \ddots & \vdots\\
    \alpha_1^{e_k} & \alpha_2^{e_k} & \cdots & \alpha_n^{e_k}\\
  \end{pmatrix}
\]

Let $J^{*} \subseteq [\ell-1]$ be the set of indices $i$ for which $j^{*} \not\in A_i$. Define
\[
  e := \sum_{i \in J^{*}} e_{\sigma_i}.
\]
We claim that the coefficient of $\alpha_{j^*}^{e}$ in the determinant expansion of (\ref{mat:RHS}) equals the determinant expansion of (\ref{mat:LHS}), up to a factor of $\pm 1$. Note this would immediately imply that (\ref{mat:LHS}) $\implies$ (\ref{mat:RHS}).

First note that if $j \in A_i$ or $\sigma_i = 0$ (i.e., $i = [\ell]$), then the $i$th row block cannot contribute any factors of $\alpha_{j^*}$. Otherwise, $i \in J^{*}$, and the maximal power of $\alpha_{j^*}$ that can be achieved in the $i$th row block is $\alpha_{j^*}^{e_{\sigma_i}}$. Thus, $e$ is an upper on the maximal power of $\alpha_{j^*}$ in any term of the determinant expansion of (\ref{mat:RHS}). If $\alpha_{j^*}^{e_{\sigma_i}}$ is achieved, it must be obtained by selecting the $\alpha_{j^*}^{e_{\sigma_i}}$ term in the $i$th row block for each $i \in J^{*}$.

Therefore, the coefficient of $\alpha_{j^*}^{e}$ can be computed by the following procedure:
\begin{itemize}
\item For all $i \in J^{*}$ replace the column $V^{T}|_{\sigma_i}$ in the $i$th row block, with a $k$-dimensional indicator vector of the $j^{*}$th row.
\end{itemize}
By applying a suitable sequences of columns swaps, we can then obtain (\ref{mat:LHS}). Thus, up to a factor of $\pm 1$, the coefficient of $\alpha_{j^*}^{e}$ in the determinant expansion of (\ref{mat:RHS}) equals the determinant expansion of (\ref{mat:LHS}), as desired.
\end{proof}

Thus, we have completed the proof of Theorem~\ref{thm:tr-gm-mds}. This also completes the proof of Theorem~\ref{thm:gm-mds-dual}.

\section{Extension to Irreducible Varieties}\label{sec:irr-var}
\subsection{Algebraic Preliminaries}
We first recall some standard facts from algebraic geometry.

Recall, an ideal $I\subseteq \F[x_1,\hdots,x_k]$ is said to be radical if $f^r\in I$ for \change{some $r \ge 1$} implies $f\in I$.
 
For an algebraically closed $\F$, a variety $X$ over $\F^k$ is defined as the zero set of a radical ideal $I\subseteq \F[x_1,\hdots,x_k]$. The Zariski topology over $X$ is generated by open sets $U_f=\{x|f(x)\ne 0,x\in X\}$ (in other words sets where a polynomial does not vanish) for all $f\in \F[x_1,\hdots,x_k]$. 

A variety $X$ is said to be irreducible if $\F[x_1,\hdots,x_n]/I$ is an integral domain or equivalently $I$ is a prime ideal. Topologically, this means that $X$ does not contain two disjoint non-empty open sets.

Any variety $X\subseteq \F^k$ is uniquely a finite union of irreducible components $X_1,\hdots,X_m$. The uniqueness means that if $X$ can be written as a union of irreducible varieties such that none of them are contained in the other then they have to be $X_1,\hdots,X_m$. Algebraically speaking there exists finitely many minimal prime ideals $I_1,\hdots,I_m$ containing $I$. Also $I_1\cap\hdots\cap I_m=I$. Note that the irreducible components can be of different dimension (for instance take the union of a curve and a hyperplane not containing the curve).

A polynomial $f$ is said to be generically non-vanishing on a variety $X$ if the open set $U_f=\{x|f(x)\ne 0,x\in X\}$ is dense for the Zariski topology on $V$. For an irreducible $V$ that simply means $f\not\in I$. Equivalently, there is at least one point in $V$ on which $f$ does not vanish. For a general $X$, it means that $f$ does vanish on each of the irreducible components of $X$.

We need a simple fact about the product of two varieties.

\begin{lemma}[Product of varieties (See 10.4.H in \cite{raviVakilFOAG})]\label{lem-prodIrred}
If $X_1\subseteq \F^{k_1}$ and $X_2\subseteq \F^{k_2}$ are two varieties then $X_1\times X_2\subseteq \F^{k_1+k_2}$ is also a variety. Furthermore, if $X_1$ and $X_2$ are irreducible then so is $X_1\times X_2$.

In general if $X_1$ and $X_2$ have irreducible components $X_{1,1},\hdots,X_{1,m_1}$ and $X_{2,1},\hdots,X_{2,m_2}$ respectively then $X_{1,i}\times X_{2,j}, i\in [m_1],j\in [m_2]$ are the irreducible components of $X_1\times X_2$.  
\end{lemma}

\begin{lemma}\label{lem-powerS}
For an algebraically closed $\F$ and an irreducible variety $X\subseteq \F^k$, there exists a $d$ ($d$ will be the dimension of $X$) such that we can find functions $f_1,\hdots,f_k$ (depending on $X$) in the formal power series ring $\F[[z_1,\hdots,z_d]]$ such that for any polynomial $g\in \F[x_1,\hdots,x_k]$ over $\F^k$ it is generically non-vanishing over $X$ if and only if $g(f_1,\hdots,f_k)\in \F[[z_1,\hdots,z_d]]$ is non-zero.

We also have that for $X^n\subseteq \F^{kn}$ any polynomial $g\in \F[x_{1,1},\hdots,x_{1,k},\hdots,x_{n,1},\hdots,x_{n,k}]$ is generically non-vanishing over $X^n$ if and only if $g(f_1(\bz_1),\hdots,f_k(\bz_1),\hdots,f_1(\bz_n),\hdots,f_k(\bz_n))$ is non-zero. 
\end{lemma}

The above follows from standard facts in algebraic geometry. We give a proof and references in the appendix.

\subsection{Algebraic GM-MDS}

\change{
\begin{definition}[$\MDS(\ell)$ property for varieties]\label{def-MDSlvar}
Let $\F$ be an algebraically closed field and $V$ be $k\times (\ell-1)k$ generic matrix (that is a matrix whose every entry is an independent formal variable). 
Given a variety $X\subseteq \F^k$ let $V^X(n)$ be a $k\times n$ matrix whose columns are generic points from $X$. 

A variety $X\subseteq \F^k$ is then said to be $\MDS(\ell)$ if for every $n$ matrix $V^X(n)$ is generically a $[n,k]-\MDS(\ell)$ code. 

Concretely, this means for any $A_1,\hdots,A_\ell\subseteq [n]$ with $|A_i| \le k$ and $|A_1|+\hdots+|A_\ell|=(\ell-1)k$, we have that $V_{A_1}\cap\hdots\cap V_{A_\ell}=0$ if and only if
  \[\det \begin{pmatrix}
    I_k & V^X(n)|_{A_1} & & & \\
    I_k & & V^X(n)|_{A_2} & &\\
    \vdots & & & \ddots &\\
    I_k & & & &  V^X(n)|_{A_{\ell}}
  \end{pmatrix},\]
  is generically non-vanishing on $X^n$.
\end{definition}
Note in the notation $V^X(n)$ we will drop the $n$ for brevity, when it is clear what value of $n$ is being used.}

A $\MDS(2)$ variety we just call $\MDS$.

We also note that being $\MDS$ equivalently means that the determinant of a $k\times k$ matrix where each column is a point in $\F^k$ is generically non-vanishing over $X^k$. This allows us to give another simple characterization of $\MDS$ varieties.

\begin{theorem}[$\MDS$ varieties are not contained in hyper-planes passing through the origin]\label{thm:MDS2VarChar}
A variety $X\subseteq \F^k$ is $\MDS$ if and only if each of its irreducible components is not contained in any hyper-plane passing through the origin.
In particular, an irreducible variety is $\MDS$ if and only if it is not contained in any hyperplane passing through the origin.
\end{theorem}
\begin{proof}
Let $f$ be the determinant of a $k\times k$ matrix $M$ where the $i$th column is a point $(x_{1,i},\hdots,x_{k,i})$ in $\F^k$ and $X$ have irreducible components $X_1,\hdots,X_m$. The irreducible components of $X^k$ are simply $X_{j_1}\times\hdots\times X_{j_k}$ with $j_1,\hdots,j_k\in [m]$. 

Say some irreducible component $X_i$ of $X$ lies in a hyperplane passing through the origin. Then $f$ vanishes on $(X_i)^k$ which means $X$ is not $\MDS$ (because $M$ has a kernel corresponding to the hyperplane).

Say $X$ is not $\MDS$ which means $f$ vanishes on some $Y=X_{j_1}\times\hdots\times X_{j_k}$. We now consider the determinants of the sub-matrices of $M$. Our goal will be to show some $X_{j_i}$ lies in a hyperplane passing through the origin which would complete the proof.

If any of the co-ordinate functions (say $x_{1,1}$) (corresponding to size $1$ sub-matrices) vanishes on $Y$ then $X_{j_1}$ lies in the co-ordinate hyperplane $x_{1,1}=0$. Say the largest-size sub-determinant which does not vanish on $Y$ has size $t>0$. This means the determinant of some $t\times t$ sub-matrix $M_1$ of $M$ does not vanish on $Y$ and all the determinants of $(t+1)\times (t+1)$ sub-matrices of $M$ vanish on $Y$. Pick a $t+1\times t+1$ sub-matrix $M_2$ of $M$ containing $M_1$. Without loss of generality we can assume the columns of $M_1$ come from $X_{j_1},\hdots,X_{j_t}$ and the extra column in $M_2$ comes from $X_{j_{t+1}}$ (the choice of rows is not important). We know that $\det(M_2)$ vanishes on $Y$. Let $R$ be the index of rows in $M_2$. Expanding the determinant of $M_2$ along the $t+1$th column will give us 
$$\det(M_2)= \sum_{i\in R} x_{i,t+1}\det(M_2'(i)),$$
where $M_2'(i)$ is the sub-matrix of $M_2$ obtained by removing the $i$th row and $t+1$the column. Note, $M_2'(t+1)=M_1$ which does not vanish on $Y$. This means it is possible to find points in $y_1\in X_{j_1},\hdots,y_t\in X_{j_t}$ such that $g=\det(M_2)(x_1=y_1,\hdots,x_t=y_2)$ is a non-zero polynomial. But as $\det(M_2)$ vanishes on $Y$ so does $g$. But that implies $X_{j_{t+1}}$ is contained in $g=0$ which is a hyperplane passing through the origin.
\end{proof}

We also generalize the definition of $\LDMDS(\le L)$ to the setting of varieties.

\begin{definition}[$\LDMDS(\le \ell)$ property for varieties]\label{def-LDMDSlvar}
\change{Let $\F$ be an algebraically closed field and $V$ be $k\times (\ell-1)k$ generic matrix (that is a matrix whose every entry is an independent formal variable).}
We say a variety $X\subseteq \F^k$ satisfies the $\LDMDS(\le \ell)$ property if and only if for ever $k\times n$ matrix \change{$V^X(n)$} \change{(as defined in Definition~\ref{def-MDSlvar})}, the dual of $V^X(n)$ is generically a $[n,n-k]$-$\MDS(\ell+1)$ code.

Concretely this means that for any $A_1,\hdots,A_{\ell+1}\subseteq [n]$ with $|A_i| \le n-k$ and $|A_1|+\hdots+|A_{\ell+1}|=\ell(n-k)$, we have that $V_{A_1}\cap\hdots\cap V_{A_{\ell+1}}=0$ if and only if \change{for every $Q\in (V^X)^\perp$
 \[\det \begin{pmatrix}
    I_{n-k} & Q|_{A_1} & & & \\
    I_{n-k} & & Q|_{A_2} & &\\
    \vdots & & & \ddots &\\
    I_{n-k} & & & &  Q|_{A_{\ell+1}}
  \end{pmatrix},\]}
is generically non-vanishing on \change{$X^{n}$.}
\end{definition}

Our main theorems for varieties are the following.

\begin{theorem}\label{thm-mdsIrred}
For \change{an} algebraic closed $\F$ if an irreducible $X$ is $\MDS$ then it is $\MDS(\ell)$ for all $\ell>2$.			
\end{theorem}

\begin{theorem}\label{thm-ldmdsIrred}
For \change{an} algebraic closed $\F$ if an irreducible $X$ is $\MDS$ then it is $\LDMDS(\le\ell)$ for all $\ell>2$.			
\end{theorem}

\subsection{Reduction to Power Series}
We want to reduce the problem to a problem about power series.

\begin{definition}[$\MDS(\ell)$-power series]
Let $V$ be $k\times (\ell-1)k$ generic matrix (that is a matrix whose every entry is an independent formal variable). Given a tuple $F=(f_1,\hdots,f_k)\in (\F[[z_1,\hdots,z_d]])^k$, let $V^F=[F(\bz_1),\hdots,F(\bz_{(\ell-1)k})]$ be a $k\times (\ell-1)k$ matrix where $\bz_1,\hdots,\bz_i=(z_{i,1},\hdots,z_{i,d}),\hdots,\bz_{(\ell-1)k}$ are distinct $d$ tuple of variables. 

$F$ is said to $\MDS(\ell)$ if and only if for any $A_1,\hdots,A_\ell\subseteq [(\ell-1)k]$ with $|A_i| \le k$ and $|A_1|+\hdots+|A_\ell|=(\ell-1)k$, we have that $V_{A_1}\cap\hdots\cap V_{A_\ell}=0$ if and only if
 \[\det \begin{pmatrix}
    I_k & V^F|_{A_1} & & & \\
    I_k & & V^F|_{A_2} & &\\
    \vdots & & & \ddots &\\
    I_k & & & &  V^F|_{A_{\ell}}
  \end{pmatrix} \neq 0,\]
\end{definition}

A $\MDS(2)$ power series we just call $\MDS$.

It is easy to show that $F=(f_1,\hdots,f_k)\in (\F[[z_1,\hdots,z_d]])^k$ is $\MDS$ if and only if

$$\Det
  \begin{pmatrix}
    f_1(\bz_1) & f_1(\bz_2) & \cdots & f_1(\bz_n)\\
    f_2(\bz_1) & f_2(\bz_2) & \cdots & f_2(\bz_n)\\
    f_3(\bz_1) & f_3(\bz_2) & \cdots & f_3(\bz_n)\\
    \vdots & \vdots & \ddots & \vdots\\
    f_k(\bz_1) & f_k(\bz_2) & \cdots & f_k(\bz_n)
  \end{pmatrix}\ne 0,$$
where $\bz_1,\hdots,\bz_i=(z_{i,1},\hdots,z_{i,d}),\hdots,\bz_k$ are distinct $d$ tuples of variables.

\begin{definition}[$\LDMDS(\le\ell)$-power series]
Let $V$ be $(n-k)\times \ell(n-k)$ generic matrix (that is a matrix whose every entry is an independent formal variable). Given a tuple $F=(f_1,\hdots,f_k)\in (\F[[z_1,\hdots,z_d]])^k$, let $V^F=[F(\bz_1),\hdots,F(\bz_{(\ell-1)k})]$ be a $k\times (\ell-1)k$ matrix where $\bz_1,\hdots,\bz_i=(z_{i,1},\hdots,z_{i,d}),\hdots,\bz_{(\ell-1)k}$ are distinct $d$ tuple of variables. 

$F$ is said to $\LDMDS(\le\ell)$ if and only if for any $A_1,\hdots,A_{\ell+1}\subseteq [\ell(n-k)]$ with $|A_i| \le n-k$ and $|A_1|+\hdots+|A_{\ell+1}|=\ell(n-k)$, we have that $V_{A_1}\cap\hdots\cap V_{A_{\ell+1}}=0$ if and only if \change{for all $Q\in (V^F)^\perp$
 \[\det \begin{pmatrix}
    I_{n-k} & Q|_{A_1} & & & \\
    I_{n-k} & & Q|_{A_2} & &\\
    \vdots & & & \ddots &\\
    I_{n-k} & & & &  Q|_{A_{\ell+1}}
  \end{pmatrix} \neq 0,\]}
\end{definition}

The main theorems about power series is the following.

\begin{theorem}\label{thm-mdsPower}
For a tuple $F=(f_1,\hdots,f_k)\in (\F[[z_1,\hdots,z_d]])^k$, $F$ is $\MDS$ if and only if $F$ is $\MDS(\ell)$ for all $\ell>2$.
\end{theorem}

\begin{theorem}\label{thm-LDmdsPower}
For \change{an} algebraic closed $\F$ if an irreducible $X$ is $\MDS$ then it is $\LDMDS(\le\ell)$ for all $\ell>2$.			
\end{theorem}

\begin{proof}[Proof of Theorem~\ref{thm-mdsIrred} assuming Theorem~\ref{thm-mdsPower}]
Recall, $\F$ is algebraically closed. By Lemma~\ref{lem-powerS} given an irreducible variety $X\subseteq \F^k$ of dimension $d$ we have elements $f_1,\hdots,f_k$ in the power series ring $\F[[z_1,\hdots,z_d]]$ such that a polynomial $g$ over $\F^{kn}$ is generically non-vanishing over $X^n$ if and only if  $$g(f_1(\bz_1),\hdots,f_k(\bz_1),\hdots,f_1(\bz_n),\hdots,f_k(\bz_n))\ne 0.$$
\end{proof}

A similar argument holds for the $\LDMDS$ theorem.

\subsection{Reduction of Power Series to Univariate Monomials}

The goal will be to reduce Theorems~\ref{thm-mdsIrred} and \ref{thm-ldmdsIrred} to equivalent theorems about curves of the form $x\rightarrow (x^{d_1},\hdots,x^{d_k}),d_1<\hdots<d_k$.

For $\mathbf{i}\in \Z_{\ge 0}^d$ we let $\bz^{\mathbf{j}}=\prod_{k=1}^d z_k^{j_k}$. If we treat $F=(f_1,\hdots,f_k)\in \F[[z_1,\hdots,z_d]]^k$ as a column vector of power series it is clear that,

$$F(\bz)=\sum_{\mathbf{j}\in \Z_{\ge 0}^d} a_{\mathbf{j}}\bz^{\mathbf{j}},$$

where $a_\mathbf{j}$ is a column vector in $\F^k$. We also define a lexicographic order on elements in $\Z_{\ge 0}^r$ and similarly a lexicographic order on monomials $\bz^{\bi},i\in \Z^{r}_{\ge 0}$.

\change{\begin{lemma}\label{lem-indepColPower}
Given a $k$-tuple of power series $F=(f_1,\hdots,f_k)\in \F[[z_1,\hdots,z_d]]^k$ if $f_1,\hdots,f_k$ are linearly independent then for
$$F(\bz)=\sum_{\mathbf{j}\in \Z_{\ge 0}^n} a_{\mathbf{j}}\bz^{\mathbf{j}}$$
satisfies two properties,
\begin{enumerate}
\item There exist $\bj^*_1<\hdots<\bj^*_k$ such that $a_{\bj^*_1},\hdots,a_{\bj^*_k}$ are linearly independent. 
\item We also have, that $a_{\bj}$ is in the $\Span(a_{\bj^*_1},\hdots,a_{\bj^*_{\beta-1}})$ for all $\bj^*_{\beta-1}<\bj<\bj^*_\beta$.
\end{enumerate}
\end{lemma}
\begin{proof}
We have

$$F(\bz)=\sum_{\mathbf{j}\in \Z_{\ge 0}^d} a_{\mathbf{j}}\bz^{\mathbf{j}},$$

with $a_{\bj}\in \F^k, \bj\in \Z_{\ge 0}^d$.

Let $\bj^*_1$ be the smallest index such that $a_{\bj^*_1}$ is non-zero. Now inductively we define $\bj^*_i$ as the smallest index such that $a_{\bj^*_i}$ does not lie in $\Span(a_{\bj^*_1},\hdots,a_{\bj^*_{i-1}})$. 

We claim that it is possible to define $\bj^*_1,\hdots,\bj^*_k$ as $f_1,\hdots,f_k$ are linearly independent. If the procedure fails then for some $i<k$, $a_{\bj}\in\Span(a_{\bj^*_1},\hdots,a_{\bj^*_{i}})$ for all $\bj\not\in \{\bj^*_1,\hdots,\bj^*_{i-1}\}$ where $a_{\bj^*_1},\hdots,a_{\bj^*_{i}}$ are linearly independent. 
If $a_{\bj^*_1},\hdots,a_{\bj^*_i}$ are column vectors then wlog let us assume the first $i$ rows are linearly independent. After an invertible linear transformation over $f_1,\hdots,f_k$ we can assume $a_{\bj^*_1},\hdots,a_{\bj^*_i}$ are elementary basis vectors but that would now imply $f_{i+1}=\hdots=f_k=0$ which gives us a contradiction. 
\end{proof}}

\change{We now give a characterization of $\MDS$ power series.}

\begin{lemma}\label{lem-powerIndep}
A $k$-tuple $F=(f_1,\hdots,f_k)\in \F[[z_1,\hdots,z_d]]^k$ is $k$-$\MDS$ if and only in $f_1,\hdots,f_k$ are linearly independent as power series over $\F$.
\end{lemma}

\begin{proof}
If $f_1,\hdots,f_k$ are linearly dependent then it is clear that $f_1,\hdots,f_k$ are not $k-\MDS$.

Let us now consider the case where $f_1,\hdots,f_k$ are linearly independent. 

By Lemma~\ref{lem-indepColPower} there exist $\bj^*_1<\hdots<\bj^*_k$ such that $a_{\bj^*_1},\hdots,a_{\bj^*_k}$ are linearly independent. Also, we have $a_{\bj}\in \Span(a_{\bj^*_1},\hdots,a_{\bj^*_{\beta-1}})$ when $\bj<\bj^*_\beta$.

To prove $F=(f_1,\hdots,f_k)$ is $\MDS$ we want to show that,
$$\Det(F(\bz_1),\hdots,F(\bz_k))=\Det
  \begin{pmatrix}
    f_1(\bz_1) & f_1(\bz_2) & \cdots & f_1(\bz_n)\\
    f_2(\bz_1) & f_2(\bz_2) & \cdots & f_2(\bz_n)\\
    f_3(\bz_1) & f_3(\bz_2) & \cdots & f_3(\bz_n)\\
    \vdots & \vdots & \ddots & \vdots\\
    f_k(\bz_1) & f_k(\bz_2) & \cdots & f_k(\bz_n)
  \end{pmatrix}\ne 0,$$
where $\bz_1,\hdots,\bz_i=(z_{i,1},\hdots,z_{i,d}),\hdots,\bz_k$ are distinct $d$ tuples of variables.

By the multilinearity of the determinant we have,

\begin{equation}\label{eq-detE1}
\Det(F(\bz_1),\hdots,F(\bz_k))= \sum_{\bj_1,\hdots,\bj_k\in \Z_{\ge 0}^d}  \Det(a_{\bj_1},\hdots,a_{\bj_k})\prod\limits_{w=1}^k \bz_w^{\bi_w}.
\end{equation}

Consider the lexicographic order over $(\bj_1,\hdots,\bj_k)\in (\Z_{\ge 0}^r)^k$. We note that the coefficient of $\prod\limits_{w=1}^k \bz_w^{\bj^*_w}$ is $\Det(a_{\bj^*_1},\hdots,a_{\bj^*_k})$ which is non-zero. We also have that this is the smallest non-zero term as well but that is not important for this proof.
\end{proof}

\begin{proof}[Proof of Theorems~\ref{thm-mdsPower} and \ref{thm-LDmdsPower} assuming Theorems~\ref{thm:gm-mds-primal} and \ref{thm:gm-mds-dual}]
We will now proceed with the reduction in two steps. First we show that it suffices to prove Theorems~\ref{thm-mdsPower} and \ref{thm-LDmdsPower} for $F=(f_1,\hdots,f_k)$ such that the $f_i$ are monomials. The second step shows that the case of multivariate monomial can be reduced to univariate ones.

\begin{claim}\label{cl:homogPower}
Theorems~\ref{thm-mdsPower} and \ref{thm-LDmdsPower} are respectively implied by Theorems~\ref{thm-mdsPower} and \ref{thm-LDmdsPower} for $F=(f_1,\hdots,f_k)$ where $f_i$ are monomials.
\end{claim} 
\begin{proof}
Given a $k$-tuple of power series $F=(f_1,\hdots,f_k)\in \F[[z_1,\hdots,z_d]]^k$ which is $\MDS$, we apply Lemma~\ref{lem-indepColPower} to find $\bj^*_1<\hdots<\bj^*_k,\bj^*_i\in \Z_{\ge 0}^d$ which satisfy the two properties in the statement of the lemma. For $\bi=(i_1,\hdots,i_t)$ in $\Z_{\ge 0}^d$ we let $|\bi|=\sum_{t=1}^d i_t$.  

Without loss of generality we can perform a change of basis over $\F^k$ and assume $a_{\bj^*_1},\hdots,a_{\bj^*_k}$ are the elementary basis vectors $e_1,\hdots,e_k$. The second property in Lemma~\ref{lem-indepColPower} then implies that the leading coefficient in $f_i(\bz)$ is $\bz^{\bj^*_i}$ for $i\in [k]$.

We see that $F'=(\bz^{\bj^*_1},\hdots,\bz^{\bj^*_k})$ is $\MDS$ and by our assumption is also $\MDS(\ell)$ and $\LDMDS(\le \ell)$ for all $\ell >0$. 

Given any $A_1,\hdots,A_\ell\subseteq [(\ell-1)k]$ with the dimension $k$ null intersection property. 
As $F'$ is $\MDS(\ell)$ we know that  

\begin{align}\label{eq:clmhomo1}
\det \begin{pmatrix}
    I_k & V^{F'}|_{A_1} & & & \\
    I_k & & V^{F'}|_{A_2} & &\\
    \vdots & & & \ddots &\\
    I_k & & & &  V^{F'}|_{A_{\ell}}
  \end{pmatrix}
\end{align}
 is non-zero.
 
We want to show that the determinant of
\begin{equation}\label{eq:clmhomo2}
\begin{pmatrix}
    I_k & V^{F}|_{A_1} & & & \\
    I_k & & V^{F}|_{A_2} & &\\
    \vdots & & & \ddots &\\
    I_k & & & &  V^{F}|_{A_{\ell}}
  \end{pmatrix}
 \end{equation}
is non-zero. The matrix in \eqref{eq:clmhomo2} has $\ell$ blocks of rows with $k$ rows in each block. Intuitively speaking we will show that the `lowest degree' part of the determinant of \eqref{eq:clmhomo2} is \eqref{eq:clmhomo1} which will complete the proof.

In \eqref{eq:clmhomo2} we multiply the $t$th column with $\beta^{|\bj^*_t|}$ for $t\in [k]$ to get the matrix.

\begin{equation}\label{eq:clmhomo3}
\begin{pmatrix}
    \text{diag}(\beta^{|\bj^*_1|},\hdots,\beta^{|\bj^*_k|}) & V^{F}|_{A_1} & & & \\
    \text{diag}(\beta^{|\bj^*_1|},\hdots,\beta^{|\bj^*_k|}) & & V^{F}|_{A_2} & &\\
    \vdots & & & \ddots &\\
    \text{diag}(\beta^{|\bj^*_1|},\hdots,\beta^{|\bj^*_k|}) & & & &  V^{F}|_{A_{\ell}}
  \end{pmatrix}
\end{equation}

In \eqref{eq:clmhomo3} we replace the variables $\bz_i=(z_{i,1},\hdots,z_{i,d}),i\in [(\ell-1)k]$ by $(\beta z_{i,1},\hdots,\beta z_{i,d})$ where $\beta$ is a extra variable. Next we divide the $t$th row in each block by $\beta^{|\bj^*_t|}$. The new matrix we get is

\begin{equation}\label{eq:clmhomo4}
\begin{pmatrix}
    I_k & V^{F''}|_{A_1} & & & \\
    I_k & & V^{F''}|_{A_2} & &\\
    \vdots & & & \ddots &\\
    I_k & & & &  V^{F''}|_{A_{\ell}}
  \end{pmatrix}
\end{equation}

where $F''(\bz)=(f_1(\beta \bz))\beta^{-|\bj^*_1|},\hdots,f_k(\beta \bz))\beta^{-|\bj^*_k|})$. We note $f_i(\beta \bz))\beta^{-|\bj^*_i|}$ is a polynomial in $\beta$ and $\bz$ for all $i$. As we only performed row and column scaling operations. The determinant of \eqref{eq:clmhomo4} is non-zero if and only if the determinant of \eqref{eq:clmhomo2} is non-zero. If we set $\beta=0$ in \eqref{eq:clmhomo4} we get the matrix in \eqref{eq:clmhomo1} which completes the proof.

The reduction for the \change{LD-MDS} theorems follows similarly.
\end{proof}

We note that the above argument would not work if the code was generated with different power series in each column. 

\begin{claim}
Theorems~\ref{thm-mdsPower} and \ref{thm-LDmdsPower} for $F=(f_1,\hdots,f_k)$ where $f_i$ are multivariate monomials are respectively implied by Theorems~\ref{thm-mdsPower} and \ref{thm-LDmdsPower} for $F=(f_1,\hdots,f_k)$ where $f_i$ are univariate monomials.
\end{claim}
\begin{proof}
Let $F=(\bz^{\bi_1},\hdots,\bz^{\bi_k}),\bz=(z_1,\hdots,z_k)$ be $\MDS$. Without loss of generality we can assume $\bi_1<\hdots<\bi_k$.   

The proof uses a standard reduction which encodes the degrees of multiple variables in the $N$-ary expansion of a number where $N$ is large enough in comparison to the individual degrees of monomials appearing in $F$. 

Take $N>|\bi_k|$.  Let $|\bi|_N=\sum_{t=1}^d i_tN^{t-1}$ for $\bi \in \Z^d$. We now substitute $z_t$ with $\alpha^{N^{t-1}}$ for $t\in [d]$ in $F$, to get $F'=(\alpha^{|\bi_1|_N},\hdots,\alpha^{|\bi_k|_N})$. By our choice of $N$ we have $|\bi_1|_N<\hdots<|\bi_k|_N$ so $F'$ is $\MDS$. By our assumption it is also $\MDS(\ell)$ and $\LDMDS(\le \ell)$.

We are done as the determinants which need to be non-vanishing for $F$ to be $\MDS(\ell)$ and $\LDMDS(\ell)$ on substituting $\bz_i=(z_{i,1},\hdots,z_{i,d})$ with $(\alpha_i^{N^0},\hdots,\alpha_i^{N^{d-1}})$ gives us the corresponding determinants for $F'$ which are non-zero by our assumption.
\end{proof}

\end{proof}

\section{\change{Conclusion}}\label{sec:openquestions}

\change{In this paper, we showed that many generic families of codes defined by polynomials or irreducible varieties are in fact higher order MDS codes. These results are established by generalizing the inductive frameworks used for showing that Reed--Solomon codes are higher order MDS codes \cite{lovett2018gmmds,yildiz2019gmmds,bgm2022} using a new concept we call `higher order MDS with a basis'. As previously mentioned, these results are used by \cite{bdg2023b} to prove that randomly punctured Algebraic-Geometric codes achieve list decoding capacity over constant-sized fields.}

We conclude the paper with some open questions.

\begin{itemize}

\item Can we generalize Theorems~\ref{thm:main-1} and \ref{thm:main-1-dual} to the setting where different columns are evaluations of different polynomial maps? More generally, can we generalize~\ref{thm:main-var} to the setting where different columns are generic points from different irreducible varieties?

\item Can we generalize our GM-MDS theorems to the setting where the columns are evaluations of $f$ at points which share some of the variables across columns. For example, linearized Reed-Solomon codes are such an example of MDS codes. See Example~\ref{ex:known-gm-mds}.

\item One very broad generalization of the GM-MDS theorem which immediately implies all known GM-MDS type theorems (including ours) is as follows. 
\begin{conjecture}[Ultimate GM-MDS conjecture]
  \label{conj:ultimategmmds}
  Let $C_0$ be any $(N,k)$-code which is MDS. Let $C$ be an $(n,k)$-code obtained by a random puncturing of $C_0$ (i.e., its generator is formed by randomly selecting $n$ columns from the generator of $C_0$). Then as with probability $1-o_N(1)$, $C$ is $\MDS(\ell)$ for all $\ell$. 
\end{conjecture}
If Conjecture~\ref{conj:ultimategmmds} is true, it says that GM-MDS is not really a special property of polynomial codes, but rather a generic property of MDS codes.

\item The study of higher order MDS codes was motivated by the study of MR Tensor Codes. Two MDS codes $C_1 \subseteq \F^m$ and $C_2 \subseteq \F^n$ form an MR tensor code if for every subset $E \subseteq [m] \times [n]$, $(C_1 \otimes C_2)|_{E}$ achieves the dimension it should if $C_1$ and $C_2$ are generically chosen. The appropriate generalization of Theorem~\ref{thm:main-1} in this setting is to show that the tensor of two generic polynomial codes is MR, although this question is open just in the case that $C_1$ and $C_2$ are generic Reed-Solomon codes.

\end{itemize}

\bibliographystyle{alpha}
\bibliography{references}

\appendix

\section{Omitted Proofs}\label{app}

\subsection{Proof of Lemma~\ref{lem:pad-mdsb}}

We shall need the following facts implicit in \cite{bgm2021mds}.

 \begin{proposition}[\cite{bgm2021mds}]\label{lem:mds7}
  Let $V$ be a $k \times n$ matrix. Let $A_1, \hdots, A_\ell \subseteq [n]$ be subsets such that $\dim(V_{A_i}) = |A_i|$. Let $H$ be a $(\ell,\ell-1)$ parity-check code. Then, the following are equivalent
  \begin{itemize}
    \item[(a)] $\dim(V_{A_1} \cap \cdots \cap V_{A_\ell}) = d$.
    \item[(b)] $\dim(\sum_{i=1}^{\ell} H_i \otimes V_{A_i}) = -d + \sum_{i=1}^{\ell} |A_i|$.
  \end{itemize}
\end{proposition}

\begin{proposition}\label{prop:sub1}
  Let $X, Y, Y', Y''$ be linear spaces such that $Y' \subseteq Y \subseteq Y''$ and each inclusion has codimension $1$. Then,
  \begin{align*}
    \dim(X \cap Y) &\ge \dim(X \cap Y') \ge \dim (X \cap Y) - 1\\
    \dim(X \cap Y) &\le \dim(X \cap Y'') \le \dim (X \cap Y) + 1
  \end{align*}
\end{proposition}
\begin{proof}
  For the first line of inequalities. Note that the first one is trivial. For the second, observe that
  \begin{align*}
    \dim(X \cap Y') &= \dim((X \cap Y) \cap Y')\\
                    &= \dim(X \cap Y) + \dim Y' - \dim (Y' + (X \cap Y))\\
                    &\ge \dim(X \cap Y) + \dim Y' - \dim Y\\
                    &= \dim(X \cap Y) - 1.
  \end{align*}
  The second inequality follows by similar logic. Or, we can use the first inequality by substituting $(Y,Y'') for (Y',Y)$).
\end{proof}

\begin{proof}[Proof of Lemma~\ref{lem:pad-mdsb}]
  We begin with part (a). Assume without loss of generality that $\sigma_1 \ge \cdots \ge \sigma_\ell$. If $\sigma_\ell \ge 1$, consider the minimal contraction $(\sigma'_1, A'_1), \hdots, (\sigma'_{\ell}, A'_\ell)$ \change{of $(\sigma_1, A_1), \hdots, (\sigma_\ell, A_\ell)$} with $\sigma'_i = \sigma_i$ and $A'_i = A_i$ for all $i \in [\ell]$ except $\sigma'_\ell = \sigma_\ell - 1$. Now for any partition $P_1 \sqcup \cdots \sqcup P_s = [\ell]$ with $\ell \in P_s$, then $\sigma'_{P_s} = \sigma'_{\ell} = \sigma_\ell-1 = \sigma_{P_s} - 1$. Thus, we have that
\begin{align*}
  \left[-n(s-1) + \sum_{i=1}^s (\sigma_{P_i} + |A_{P_i}|)\right] &= \left[-n(s-1) + \sum_{i=1}^{s} (\sigma'_{P_i} + |A'_{P_i}|)\right] - 1.\end{align*}
Therefore, by taking the maximum over all such partitions and invoking Proposition~\ref{prop:MDSb-inter-form}, we have that $d-1 = \gid_k((\sigma'_1, A'_1), \hdots, (\sigma'_\ell, A'_\ell)).$

Now assume that $\sigma_\ell = 0$ and consider $W \in \F^{k \times b}$ and $W' \in \F^{k \times n}$ generic. Let $X := \bigcap_{i=1}^\ell (W_{\le \sigma_i} + W'_{A_i})$. Since $d \ge 1$, we have that $\change{X} \neq 0$. Since $X \subseteq W'_{A_\ell}$ and $|A_\ell| \le k$, each nonzero vector $x \in X$ is a unique nonzero linear combination of the columns of $W'|_{A_{\ell}}$. This implies there is at least one index $i^* \in A_{\ell}$ such that some vector $x \in X$ requires $i^*$ in its linear combination. If we let $A'_{\ell} = A_{\ell} \setminus \{i^*\}$ and otherwise let $\sigma'_i = \sigma_i$ and $A'_i = A_i$, then $\dim(X \cap W'_{A'_{\ell}}) \le \dim(X)-1 = d-1$. Likewise by Proposition~\ref{prop:MDSb-inter-form} we must have that 
\[
  \gid_k((\sigma'_1, A'_1), \hdots, (\sigma'_\ell, A'_\ell)) = \dim(X \cap W'_{A'_{\ell}}) = d-1.
\]

We now move on to part (b). Let $h_1, \hdots, h_\ell \in \F^{\ell-1}$ be nonzero vectors forming an $(\ell, \ell-1)$-MDS code. By Proposition~\ref{lem:mds7} (and adapting it like in Proposition~\ref{prop:MDSb-inter-form}), we have that $\gid_k((\sigma_1, A_1), \hdots, (\sigma_\ell, A_\ell)) = 0$ if and only if 
\[
  \dim \sum_{i=1}^{\ell} h_i \otimes (W_{\le \sigma_i} + W'_{A_i}) = \sum_{i=1}^{\ell} \sigma_i + |A_i|.
\]
Since $n \ge k$, we also have that
\[
  \dim \sum_{i=1}^{\ell} h_i \otimes (W_{\le \sigma_i} + W'_{[n]}) = \dim (\F^{\ell-1} \otimes \F^k) = (\ell-1)k.
\]
Thus, since $\sum_{i=1}^{\ell} \sigma_i + |A_i| < (\ell-1)k$, we argue there exists $i \in [\ell]$ and $j \in [n]$ with $j \not\in A_i$ such that 
\[
\dim \left[h_i \otimes W'_j + \sum_{i=1}^{\ell} h_i \otimes (W_{\le \sigma_i} + W'_{A_i}) \right] = 1 + \dim \sum_{i=1}^{\ell} h_i \otimes (W_{\le \sigma_i} + W'_{A_i}).
\]
If no such $h_i \otimes W'_j$ exists, then every $h_i \otimes W'_j$ must be a linear combination of vectors in $h_i \otimes (W_{\le \sigma_i} + W'_{A_i})$, implying that  $\dim \sum_{i=1}^{\ell} h_i \otimes (W_{\le \sigma_i} + W'_{A_i}) = \dim \sum_{i=1}^{\ell} h_i \otimes (W_{\le \sigma_i} + W'_{[n]})$, a contradiction.

Let $(\sigma'_1, A'_1), \hdots, (\sigma'_\ell, A'_\ell)$ be the minimal expansion \change{of $(\sigma_1, A_1), \hdots, (\sigma_\ell, A_\ell)$} for which $j$ is added to set $A_i$. We then have that $0 = \gid_k((\sigma'_1, A'_1), \hdots, (\sigma'_\ell, A'_\ell)),$ as desired.
\end{proof}

\subsection{Proof of Lemma~\ref{lem-powerS}}

$\F$ throughout is algebraically closed.

At a high level the idea is that we can take a local view at a single point on the irreducible variety and that will give us a power series expansion.

\begin{definition}[Stalk of a variety at a point]
Given a variety $X\subseteq \F^k$ cut out by a radical ideal $I\subseteq \F[x_1,\hdots,x_k]$ and a point $x\in X$ the stalk of $X$ at $x$ is constructed by taking the ring $\F[x_1,\hdots,x_k]/I$ and inverting every polynomial $f$ which does not vanish on $x$.
\end{definition}

The stalk corresponds to a local view of the variety $X$ around $x$.

By construction for a variety $X$ cut out by $I$ there is a map from $\F[x_1,\hdots,x_k]/I$ to the stalk at a point $x\in X$. It is easy to see that this map is injective if $X$ is irreducible (equivalently $I$ is prime). This means, for an irreducible variety looking at the stalk is enough to check if a polynomial is generically non-vanishing.

We next need the notion of smoothness. 

\begin{definition}[Smooth point]
A point $x$ on a variety $X\subseteq \F^k$ of dimension $d$ cut out by a radical ideal $I=\langle f_1,\hdots,f_m\rangle$ is said to be smooth when the Jacobian of the polynomials $f_1,\hdots,f_m$ is of rank $k-d$ at $x$.
\end{definition}

The stalks of smooth points have very nice properties. We will need the following one.

\begin{theorem}[Stalk at a smooth point (See 13.2.J and Section 28.2 in \cite{raviVakilFOAG})]\label{alem-stalkCompletion}
Given a smooth point $x$ on an irreducible variety $X\subseteq \F^k$ of dimension $d$, the stalk at $x$ injects into the power series ring $\F[[z_1,\hdots,z_d]]$.
\end{theorem}

A priori it is not clear if every variety will have a smooth points. Almost all points on a variety are smooth.

\begin{theorem}[Generic Smoothness of varieties (see 21.6.1 in \cite{raviVakilFOAG})]\label{alem-genSmooth}
If $X$ is an irreducible variety then there is a dense open subset of $X$ which only contains smooth points.
\end{theorem}

We are now ready to prove Lemma~\ref{lem-powerS}.

\begin{proof}
Any irreducible variety will have at least one smooth point on it as the set of smooth points is dense (Lemma~\ref{alem-genSmooth}). By Lemma~\ref{alem-stalkCompletion} the stalk at a smooth point on a dimension $d$ irreducible variety maps injectively into $\F[[z_1,\hdots,z_d]]$. This means the coordinate functions $x_1,\hdots,x_k$ on the stalk map to power series $f_1,\hdots,f_k$.

For an irreducible variety checking a function is generically non-vanishing is equivalent to checking if the function is non-zero on a stalk of any point. As the stalk maps injectively into the power series ring we are done with the first part of the theorem.

The second part simply follows from the fact that the product of irreducible varieties is irreducible and we can take the product of the stalk of a point on $X$ to get a stalk on $X^n$.
\end{proof}

\end{document}